\documentclass[12pt,a4paper]{article}

\usepackage{amsmath,amsthm,amssymb,latexsym,amsfonts,mathabx,amscd,bbm,a4wide}

\usepackage{dsfont}
\usepackage{mathrsfs}
\usepackage{setspace}
\usepackage{hyperref}
\usepackage{color}
\usepackage{mathtools}
\usepackage{comment}
\usepackage{changes}

\allowdisplaybreaks
\newcommand{\ud}{\mathrm{d}}

\numberwithin{equation}{section}

\newtheorem{theorem}{Theorem}[section]
\newtheorem{lemma}[theorem]{Lemma}

\newtheorem{cor}[theorem]{Corollary}
\newtheorem{remark}[theorem]{Remark}
\newtheorem{remarks}[theorem]{Remarks}

\theoremstyle{definition}
\newtheorem{defn}[theorem]{Definition}

\newcommand{\E}{\mathds{E}}

\numberwithin{equation}{section}

\begin{document}

\thispagestyle{empty}

\vspace*{1cm}

\begin{center}
	
	{\Large \bf  On Bose--Einstein condensation in the Luttinger--Sy model 
	with finite interaction strength} \\

	\vspace*{2cm}
	
	{\large  Joachim~Kerner \footnote{E-mail address: {\tt joachim.kerner@fernuni-hagen.de}}, Maximilian~Pechmann \footnote{E-mail address: {\tt maximilian.pechmann@fernuni-hagen.de}}, and Wolfgang~Spitzer \footnote{E-mail address: {\tt wolfgang.spitzer@fernuni-hagen.de}}}%
	
	\vspace*{5mm}
	
	Department of Mathematics and Computer Science\\
	FernUniversit\"{a}t in Hagen\\
	58084 Hagen\\
	Germany\\
	
\end{center}

\vfill

\begin{abstract} We study Bose--Einstein condensation (BEC) in the Luttinger--Sy model. Here, Bose point particles in one spatial dimension do not interact with each other, but, through a positive (repulsive) point potential with impurities which are randomly located along the real line according to the points of a Poisson process. Our emphasis is on the case in which the interaction strength is not infinite. As a main result, we prove that in thermal equilibrium the one-particle ground state is macroscopically occupied, provided that the particle density is larger than a critical one depending on the temperature. 
\end{abstract}

\newpage

%
\section{Introduction}
In this paper, we study Bose--Einstein condensation (BEC)~\cite{PaperEinstein1,PaperEinstein2} in the Luttinger--Sy (LS) model with finite interaction strength~\cite{luttinger1973low,luttinger1973bose}. Conventional BEC refers to a macroscopic occupation of a one-particle ground state. If the particles do not interact with each other (ideal Bose gas), this one-particle state is the ground state of the one-particle Hamiltonian~\cite{PO56}. In contrast to this, generalized BEC requires a macroscopic occupation of an arbitrarily narrow energy interval of one-particle states~\cite{casimir1968bose, van1982generalized, van1983condensation, van1986general, van1986general2,ZagBru,girardeau1960relationship, jaeck2010nature}. More specifically, generalized BEC is then classified into three different types: Type-I or type-II BEC is present whenever finitely or infinitely many one-particle states in this narrow energy interval are macroscopically occupied; generalized BEC without any one-particle state in this energy interval being macroscopically occupied is defined as type-III.

In an ideal Bose gas, to prove generalized BEC (which is done by showing that some critical density is finite) is much simpler than proving macroscopic occupation of the one-particle ground state since for this one has to estimate the gaps between consecutive eigenvalues of the one-particle Hamiltonian~\cite{LenobleZagrebnovLuttingerSy, zagrebnov2007bose, jaeck2010nature}. Note here that external potentials may lead to a generalized BEC by altering the density of states \cite{lenoble2004bose}. However, only very limited results regarding the type of BEC in random external potentials are available~\cite{jaeck2010nature}. Actually, to the best of our knowledge, the type of BEC has so far been determined rigorously for the LS model with infinite interaction strength only, where it is of type-I~\cite{LenobleZagrebnovLuttingerSy, zagrebnov2007bose}. 

The LS model is a one-dimensional, continuous model with a random Hamiltonian which is the sum of kinetic energy (described by the Laplacian) and a random point potential. Here, a Poisson process on the real line generates a sequence of points. At all these points we attach a $\delta$-function of mass $\gamma>0$. This $\gamma$ has the meaning of the interaction strength. Formally setting $\gamma=\infty$ one arrives at the LS model with infinite interaction strength, which has been investigated in~\cite{luttinger1973low,luttinger1973bose,LenobleZagrebnovLuttingerSy, zagrebnov2007bose}, see also~\cite{KPS18}.

%
%

The paper is organized as follows: In Section~\ref{SecPrelim} we formulate the model in terms of its Hamiltonian which itself is obtained from an associated quadratic form. We recall well-known facts regarding the (limiting) integrated density of states, define (generalized) BEC and state the known theorem on the existence of generalized BEC. In Section~\ref{secMainResults} we then present our main result, Theorem~\ref{LSM ns main theorem}, which leads to a proof of almost sure macroscopic occupation of the one-particle ground state, see Theorem~\ref{MacroscopicOccupationGroundstate}.

Auxiliary results are summarized in the appendix, to which we refer throughout the manuscript.

%
\section{Preliminaries}\label{SecPrelim}
In order to introduce the Luttinger--Sy model with finite (repulsive) interaction (referred to as LS model in the rest of the paper) one starts with a Poisson point process $X$ on $\mathds{R}$ with intensity $\nu > 0$ on some probability space $(\Omega,\mathcal{F},\mathds{P})$. More explicitly, $\mathds{P}$-almost surely $X(\omega)=\{x_j(\omega): j \in \mathbb{Z}\}$ is a strictly increasing sequence of points (also called atoms) $x_j=x_j(\omega) \in \mathds{R}$ and zero is contained in the interval $(x_0(\omega),x_1(\omega))$. The probability that a given bounded Borel subset $\Lambda \subset \mathds{R}$ contains exactly $m$ points is
\begin{align}
\mathds{P}\left(\omega:\left|X(\omega)\cap \Lambda)\right|=m\right)=\frac{(\nu|\Lambda|)^m}{m!}\mathrm{e}^{-\nu |\Lambda|}\ , \quad m \in \mathds{N}_0\ .
\end{align}
Furthermore, for two disjoint Borel sets $\Lambda_1$ and $\Lambda_2$, the events $\{\omega:\left|X(\omega)\cap \Lambda_1)\right|=m_1\}$ and $\{\omega:\left|X(\omega)\cap \Lambda_2)\right|=m_2\}$ are stochastically independent. Note that, depending on the context, $|\cdot|$ refers to the Lebesgue measure or to the counting measure.

In a next step one places at each atom $x_j(\omega)$ a $\delta$-distribution of mass $\gamma>0$, also called the interaction strength. The underlying one-particle Hamiltonian is informally given by 
\begin{align}\label{FullHamiltonianRealLine}\begin{split}
\mathrm{h}_{\gamma}(\omega)&:=-\frac{\ud^2}{\ud x^2}+V_{\gamma}(\omega,\cdot)\ ,
\end{split}
\end{align}
with (external) potential
\begin{align}
V_{\gamma}(\omega,\cdot):=\gamma \sum_{j \in \mathds Z}\delta(\cdot-x_j(\omega))\ .
\end{align}
Note that a rigorous definition of~\eqref{FullHamiltonianRealLine} can be obtained via the construction of a suitable quadratic form on the Hilbert space $L^2(\mathds{R})$. However, since we are interested in studying BEC, we have to employ a thermodynamic limit which, in particular, means that we have to restrict the one-particle configuration space from $\mathds{R}$ to the bounded interval $\Lambda_N:=(-L_N/2,L_N/2)$ with $L_N:=N/\rho$; here $\rho > 0$ denotes the particle density and $N \in \mathds{N}$ is a scaling parameter (the particle number) which eventually will go to infinity. Consequently, we introduce $\mathrm{h}^{N}_{\gamma}(\omega)$ as the finite-volume version of $\mathrm{h}_{\gamma}(\omega)$, defined on the Hilbert space $L^2(\Lambda_N)$. More explicitly, for all $N$ and $\mathds{P}$-almost all $\omega \in \Omega$, one defines the (quadratic) form
\begin{align}
q^{N}_{\gamma}(\omega)[\varphi]:=\int\limits_{\Lambda_N}|\varphi^{\prime}(x)|^2\ \ud x + \gamma \sum_{j: x_j(\omega) \in \Lambda_N}|\varphi(x_j(\omega))|^2
\end{align}
on $L^2(\Lambda_N)$ with (Dirichlet) form domain $H^1_0(\Lambda_N) :=\{\varphi \in L^2(\Lambda_N): \varphi^{\prime} \in L^2(\Lambda_N), \varphi(-L_N/2)$ $ = \varphi(+L_N/2)=0  \}$. This form is positive, densely defined and closed. Hence, due to the representation theorem for quadratic forms, there exists a unique self-adjoint operator associated with this form. Informally, this operator is 
\begin{align}\label{FormalOperatorLSModel}
\mathrm{h}^{N}_{\gamma}(\omega):=-\frac{\ud^2}{\ud x^2}+\gamma \sum_{j:x_{j}(\omega) \in \Lambda_N}\delta(\cdot-x_j(\omega))\ .
\end{align}
The spectrum of $\mathrm{h}^{N}_{\gamma}(\omega)$ is $\mathds{P}$-almost surely purely discrete. 
%
%
We write $(E^{j,\omega}_N)_{j=1}^{\infty}$ for the eigenvalues of $\mathrm{h}^{N}_{\gamma}(\omega)$, ordered in increasing order, i.e., $0 < E^{1,\omega}_N < E^{2,\omega}_N < E^{3,\omega}_N < ...$, taking into account that the eigenvalues are $\mathds{P}$-almost surely non-degenerate~\cite{kirsch1985universal}. Also, we denote the eigenfunctions corresponding to $(E^{j,\omega}_N)_{j=1}^{\infty}$ as $(\varphi^{j,\omega}_N)_{j=1}^{\infty} \subset L^2(\Lambda_N)$.

We define the integrated density of states $\mathcal{N}^{\mathrm{I},\omega}_{N}:\mathds{R} \rightarrow \mathds{N}_0$ and the density of states (measure\footnote{When we speak of the density of states we always mean the corresponding measure and never the density in the sense of the Radon--Nikodym derivative of this measure with respect to the Lebesgue measure. We do not even know whether the latter exists in this model.}) $\mathcal{N}^{\omega}_{N}$ associated with $\mathrm{h}^{N}_{\gamma}(\omega)$, for $\mathds{P}$-almost all $\omega \in \Omega$ and all $N \in \mathds{N}$, via
\begin{align}  \label{definition finite volume integrated density of states}
\mathcal N_N^{\mathrm{I},\omega}(E) := \int\limits_{(-\infty,E)} \mathcal N_N^{\omega} (\mathrm{d} \tilde{E}) := \dfrac{1}{|\Lambda_N|} \left| \left\{ j : E_N^{j,\omega} < E \right\} \right|\ , \quad E \in \mathds{R}\ .
\end{align}
Note that $\mathcal N_N^{\mathrm{I},\omega}(\cdot)$ is non-decreasing and left-continuous.

In the grand-canonical ensemble, the (mean) number of particles $n_N^{j,\omega}$ occupying the eigenstate $\varphi^{j,\omega}_N$ at inverse temperature $\beta \in (0,\infty)$ is given by
\begin{align}
n_N^{j,\omega} = \left( \mathrm{e}^{\beta (E_N^{j,\omega} - \mu_N^{\omega})} - 1 \right)^{-1}\ ,
\end{align}
for $\mathds{P}$-almost all $\omega \in \Omega$. Here, $\mu_N^{\omega} \in (-\infty,E^{1,\omega}_{N})$ denotes the chemical potential, see \cite[Section 5.2.5]{BratteliRobinson}, \cite{LanWil79}. The chemical potential $\mu_N^{\omega}$ is such that
\begin{align} \label{Gleichung Bedingung mu}
\sum_{j=1}^{\infty}  n_N^{j,\omega}  = N 
\end{align}
holds for $\mathds{P}$-almost all $\omega \in \Omega$ and all $N \in \mathds{N}$. Note that at fixed $\rho$ and $\beta$, $\mu_N^{\omega}$ is uniquely determined. Also, \eqref{Gleichung Bedingung mu} can be written equivalently as
\begin{align} \label{Gleichung Bedingung mu aequivalent}
\int\limits_{\mathds{R}} \left( \mathrm{e}^{\beta (E - \mu_N^{\omega})} - 1 \right)^{-1} \mathcal N_N^{\omega} (\mathrm{d} E) =\rho\ .
\end{align}
To simplify notation later on we introduce the Bose function $\mathcal{B}:\mathds R  \rightarrow \mathds{R}$, $$\mathcal{B}(E):=\left( \mathrm{e}^{\beta E} - 1 \right)^{-1} \Theta(E) \ , $$
where $\Theta$ is the indicator function on $(0,\infty)$.
\begin{defn}[Thermodynamic limit] For fixed $\rho,\beta > 0$, the thermodynamic limit is realized as the limit $N \rightarrow \infty$ together with $L_N:=N/\rho$ and $\mu_N^{\omega}$ such that~\eqref{Gleichung Bedingung mu} holds for $\mathds{P}$-almost all $\omega \in \Omega$ and all $N \in \mathds{N}$. 
\end{defn}
\begin{defn}[Bose--Einstein condensation]\label{Definition makroskopische Besetzung}
	For fixed $\rho,\beta > 0$, we say that the $j$th eigenstate is $\mathds{P}$-almost surely \textit{macroscopically occupied} (in the thermal equilibrium state characterized by $\rho$ and $\beta$) if
	\begin{align} \label{Def m o}
	\mathds P \left(\omega: \limsup\limits_{N \to \infty} \dfrac{n_N^{j,\omega}}{N}   > 0 \right) = 1 \ .
	\end{align}
	Moreover, \textit{generalized BEC} is said to occur $\mathds{P}$-almost surely if an arbitrarily narrow energy interval at the lower edge of the spectrum is $\mathds{P}$-almost surely macroscopically occupied, i.e., if 
	\begin{align}
	\mathds P \left(\omega:\lim\limits_{\epsilon \searrow 0} \limsup\limits_{N \to \infty} \dfrac{1}{N} \sum\limits_{j : E_N^{j,\omega}-E_N^{1,\omega} \le \epsilon}   n_N^{j,\omega}  > 0 \right) = 1 \ . 
	\end{align}
	\end{defn}
\begin{remarks} \begin{enumerate} \item[(i)] In Definition \ref{Definition makroskopische Besetzung} there is not a single fixed state that is macroscopically occupied but there is a sequence of eigenstates $(\varphi^{j,\omega}_N)_{j=1}^{\infty}$ with the property \eqref{Def m o}. This is a common abuse of language, which we adopt. 

\item[(ii)] One could argue to call the $j$th eigenstate $\mathds{P}$-almost surely \emph{not} macroscopically occupied if $\mathds{P}$-almost surely $n_N^{j,\omega}/N$ converges to zero as $N\to\infty$ and then the $j$th eigenstate $\mathds{P}$-almost surely macroscopically occupied if $\mathds{P}$-almost surely it does not converge to zero. The latter condition is \eqref{Def m o}. Still it would be nice to replace the limit superior in~\eqref{Def m o} by the limit inferior or the limit if it existed. In any case, our result implies that $\mathds{P}$-almost surely for some subsequence of intervals the fraction of particles occupying the ground state converges to some (strictly) positive value.

\item[(iii)] One may replace the $\mathds P$-almost sure property in \eqref{Def m o} by the condition that the expectation $\mathds E[n_N^{j,\omega}/N]$ is strictly positive, uniformly in $N$. See Remark \ref{remark 3.6}.
\end{enumerate}
\end{remarks}
%

From Definition~\ref{Definition makroskopische Besetzung} it is clear that a $\mathds{P}$-almost sure macroscopic occupation of the one-particle ground state implies $\mathds{P}$-almost sure generalized BEC. However, the converse does not always hold, see, e.g., \cite{van1982generalized} or \cite{LenobleZagrebnovLuttingerSy}. 


In order to prove generalized BEC for the LS model, one makes use of the fact that the limiting integrated density of states $\mathcal N_{\infty}^{\mathrm{I}}(\cdot)=\lim_{N \rightarrow \infty}\mathcal N_N^{\mathrm{I},\omega}(\cdot)$ $\mathds{P}$-almost surely (the limit being in the vague sense) exhibits a Lifshitz-tail behavior at the bottom of the spectrum. We write $\mathcal N_{\infty}^{\mathrm{I}}(E):=\int_{(-\infty,E)} \mathcal N_{\infty} (\mathrm{d} E)$ and call $\mathcal N_{\infty}$ the limiting density of states (measure). 

Alternatively, see for instance \cite[(3.7)]{HupferLeschkeDensityofStates}, one can define $\mathcal N_{\infty}^{\mathrm{I}}$ directly for the infinite-volume Hamiltonian $h_\gamma(\omega)$ on $\mathds R$. To this end, let $\Theta$ be the indicator function on $(0,\infty)$ and for $x,y\in\mathds{R}$ let $\Theta(E - h_\gamma(\omega))(x,y)$ be the kernel of the spectral projection of $h_\gamma(\omega)$ onto the eigenspace with eigenvalues (strictly) less than $E$. Then, for all $E \in \mathds R$,
$$ \mathcal N_\infty^{\mathrm{I}}(E) = \mathds E\big[\Theta(E-h_\gamma(\omega))(0,0)\big]\,.
$$ 
This function on the right-hand side has the desired properties of being non-negative, left-continuous, and non-decreasing. Moreover, $\mathcal N_N^{\mathrm{I},\omega}$ converges to this function vaguely $\mathds{P}$-almost surely as $N\to\infty$.


%
 \begin{theorem} {\cite{EggarterSomeExact,GredeskulPastur75},\cite[Theorems 5.29,~6.7]{pastur1992spectra}} \label{Lifshitz Auslaufer one dimensional} 
 	For the LS model one has
 	\begin{align}
 	\mathcal N_{\infty}^{\mathrm{I}}(E) = \exp \left( - \pi \nu E^{-1/2} \big[1 + O(E^{1/2}) \big] \right)
 	\end{align}
 	as $E \searrow 0$, i.e., there exists a constant $\widetilde M := \widetilde M(\gamma) > 0$ and an $\widetilde E > 0$ such that for all $0 < E \le \widetilde E$  we have
 	\begin{align}
 	\mathcal N_{\infty}^{\mathrm{I}}(E) \le \widetilde M \mathrm{e}^{-\nu \pi E^{-1/2}}  \ .
 	\end{align}
 	In addition, there exists a constant $c > 0$ such that for all $E\ge0$
 	\begin{align}\label{UpperEstimateIntegratedDensityStates}
 	\mathcal N_{\infty}^{\mathrm{I}}(E) \leq cE^{1/2}\ .
 	\end{align}
 \end{theorem}
\begin{remark} In the proof of Lemma \ref{Lemma wegen vager statt schwacher Konvergenz} we use the estimate $\mathcal N_{N}^{\mathrm{I,\omega}}(E) \leq  \mathcal N_N^{\mathrm{I},(0)}(E)\leq \pi^{-1}\sqrt{E}$ for $\mathds P$-almost all $\omega \in \Omega$, all $E \ge 0$ and all $N \in \mathds N$ where $\mathcal N_N^{\mathrm{I},(0)}$ denotes the integrated density of states of the free Hamiltonian $-\ud^2/ \ud x^2$ on $H^1_0(\Lambda_N)$.
\end{remark}
Theorem~\ref{Lifshitz Auslaufer one dimensional} implies that the critical density 
\begin{equation}
\rho_{c}(\beta):=\sup_{\mu \in (-\infty,0)}\left\{\int\limits_{\mathds{R}} \mathcal{B}(E-\mu)\ \mathcal N_{\infty} (\ud E) \right\}\label{EquationCriticalDensity}
\,=\,\int_{(0,\infty)}\mathcal{B}(E)\ \mathcal N_{\infty} (\ud E)
\end{equation}
is finite, i.e., $\rho_{c}(\beta) < \infty$; see Lemma~\ref{finite critical density in LSmodel}.
\begin{remark}\label{RemarkSubsequenceChemicalPotentials} Whenever $\rho \geq \rho_{c}(\beta)$ then the sequence $\mu^{\omega}_N$ of chemical potentials converges $\mathds P$-almost surely to zero and to a strictly negative value whenever $\rho < \rho_c(\beta)$, see Lemma~\ref{Satz Konvergenzverhalten mu} and~\cite{lenoble2004bose}. In this context note that the infimum of the spectrum of  $\mathrm{h}_{\gamma}(\omega)$ is zero.
	\end{remark}
For the LS model one can then prove the following statement.
\begin{theorem}{\cite{lenoble2004bose}}\label{TheoremGeneralizedBEC} Generalized BEC in the LS model occurs $\mathds P$-almost surely if and only if $\rho > \rho_{c}(\beta)$.
\end{theorem}
%

%
\section{Main results}\label{secMainResults}
Since our goal is to prove macroscopic occupation of the one-particle ground state it is necessary to control the rate of convergence of the energies of excited states which also converge to zero.
\begin{remark}\label{RemarkTwoSets} We repeatedly use the following inequality: If $\Omega_1,\Omega_2 \subset \Omega$ are two events with $\mathds{P}(\Omega_j)\geq 1-\eta_j$ then 
	\begin{align*}
	\mathds{P}(\Omega_1 \cap \Omega_2) \geq 1 - \eta_1-\eta_2\ .
	\end{align*}
\end{remark}
\begin{theorem}[Energy gap] \label{satz E2}
	There exists a constant $M > 0$ with the following property: For any $0 < \eta < 2$ and any $c_2 > 2$ there exists an $\widetilde N = \widetilde N(\eta, c_2) \in \mathds N$ such that for all $N \ge \widetilde N$ we have $\mathds P(\Omega_2(\eta,c_2)) > 1 - 5\eta/8$ where
	\begin{align*}
	\Omega_2(\eta,c_2) := \left\{ \omega: E_N^{1,\omega}  \le  \left( \dfrac{\pi\nu}{\ln (c_1 L_N)} \right)^2 \text{and}\ E_N^{c_3,\omega} \ge \left( \dfrac{\pi\nu}{\ln (c_1 L_N) - \ln(c_2/2)} \right)^2 \right\} \ ,
	\end{align*}
	$c_1 := - \nu / [4 \ln(\eta / 2)]$, and $c_3 := \lceil 4M  c_2/ (\eta c_1) \rceil + 1$.
\end{theorem} 
\begin{proof} In a first step we note that, by~\cite[Theorem~C.6]{KPS18} and the Rayleigh--Ritz variational principle, for $0 < \eta < 2$ there exists an $\widetilde N = \widetilde N(\eta) \in \mathds N$ such that 
	\begin{align*}
	\mathds P \left( \omega: E_N^{1,\omega} \le  \left( \dfrac{\pi\nu}{\ln(c_1L_N) } \right)^2  \right) > 1 - \dfrac{1}{2}\eta
	\end{align*}
	for all $N \ge \widetilde N$ with $c_1 := - \nu / [4 \ln(\eta/ 2)]$. 
	
	 We set
	 \begin{align}\label{TildeEnergie}
	 \widetilde E_N :=  \left( \dfrac{\pi\nu}{\ln (c_1 L_N) - \ln(c_2/2)} \right)^2 \ .
	 \end{align}
   	 Then, according to Theorem~\ref{Lifshitz Auslaufer one dimensional} and Lemma~\ref{lemma E2} (identifying $\widetilde{M}=M^{1/2}$) there exists an $M > 0$ and an $\widetilde N_1 \in \mathds N$ such that, with $\mathcal E:=(\pi\nu )^{-1}\ln(M^{1/2})$,
	 \begin{align*}
	 & \E \left[ \mathcal N_N^{\mathrm{I}, \omega} \left( \widetilde E_N \right) \right] \le \mathcal N_{\infty}^{\mathrm{I}} \left( \left(\widetilde E_N^{-1/2} - \mathcal E \right)^{-2} \right) \le M  \exp \left( - \nu \pi \widetilde E_N^{-1/2} \right)
	 = M \dfrac{c_2}{2 c_1L_N}
	 \end{align*}
	 and consequently
	 \begin{align*}
	 \E \left[ \left| \left\{ j : E_N^{j,\omega} \le \widetilde E_N \right\} \right| \right] = & L_N \cdot \E \left[ \mathcal N_N^{\mathrm{I},\omega} \left( \widetilde E_N \right) \right]  \le M \dfrac{c_2}{2 c_1}
	 \end{align*}
	 for all $N \ge \widetilde N_1$. Therefore,
	 \begin{align*}
	 k \, \mathds P \left( \omega: \left| \left\{ j : E_N^{j,\omega} \le \widetilde E_N \right\} \right| \ge k \right) & \le \E \left[ \left| \left\{ j : E_N^{j,\omega} \le \widetilde E_N \right\} \right|  \right] \le M \dfrac{c_2}{2 c_1}
	 \end{align*}
	 and 
	 \begin{align*}
	 \mathds P \left( \omega: E_N^{k + 1,\omega} \ge \widetilde E_N \right) = \mathds P \left( \omega: \left| \left\{ j : E_N^{j,\omega} \le \widetilde E_N \right\} \right| \le k \right) & \ge 1 - M \dfrac{c_2}{2 c_1 k }
	 \end{align*}
	 for all $k \in \mathds N$ and all $N \ge \widetilde N_1$. Setting $k = c_3 - 1$ now finishes this proof by taking Remark~\ref{RemarkTwoSets} into account.
\end{proof}
For the following statement we introduce, for fixed $\rho,\beta$ and $\mathds{P}$-almost all $\omega \in \Omega$,
\begin{align}\label{DensityRhoZero}\begin{split}
\rho_0(\beta):=\lim_{\epsilon \searrow 0}\liminf_{N \rightarrow \infty}\int\limits_{(0,\epsilon]} \mathcal{B}(E - \mu_N^{\omega}) \mathcal N_N^{\omega} (\mathrm{d} E) \ .
\end{split}
\end{align}
Note that the right-hand side is $\mathds{P}$-almost surely equal to a non-random function of $\beta$; see Lemma~\ref{Lemma beweis lim int epsilon infty und lim lim int epsilion infty rhoc vage Konvergenz} for details. Moreover, this lemma implies $\rho_0(\beta) = \rho-\rho_{c}(\beta)$ whenever $\rho > \rho_{c}(\beta)$.
\begin{theorem}[Macroscopic occupation in probability] \label{LSM ns main theorem}
	Assume that $\rho > \rho_c(\beta)$. Then there exists an $M > 0$ with the following property: For all $0 < \eta \le 1/2$ and all $c_2 > 2$ there exists an $\widetilde N = \widetilde N(\eta, c_2) \in \mathds N$ such that for all $N \ge \widetilde N$,
	\begin{align*}
	\mathds P \left(\omega: \dfrac{n_N^{1,\omega}}{N}  \ge \dfrac{1}{c_3} \rho^{-1}(1 - \eta^{1/2})(1 - \eta) \rho_0(\beta) \right) \ge 1 - 4 \dfrac{\rho_0(\beta) + \rho + 1}{\rho_0(\beta)} \eta^{1/2} - 6\eta/8
	\end{align*}
	with $c_3 = c_3(\eta, c_2, M) = \lceil 4M  c_2/ (\eta c_1) \rceil + 1 $ and $c_1$ as in~Theorem~\ref{satz E2}.
\end{theorem}
\begin{proof}
	For $0 < \eta \le 1/2$
		we define, for $N \in \mathds N$,
	\begin{align*}
	\Omega_3(\eta,c_2) & := \left\{ \omega: \left(\dfrac{\pi \nu}{3 \ln (L_N)} \right)^2 \le E_N^{1,\omega} \le \left(\dfrac{\pi\nu}{\ln (c_1 L_N)}\right)^2\text{ and } E_N^{c_3,\omega} \ge \widetilde E_N \right\}
	\end{align*}
	with $c_1,c_3$ as in Theorem~\ref{satz E2} and $\widetilde E_N$ as in \eqref{TildeEnergie}.
	
	According to Theorem~\ref{satz E2} and Lemma~\ref{lower bound ground state energy}, there exists an $\widetilde N = \widetilde N(\eta, c_2) \in \mathds N$ such that for all $N \ge \widetilde N$ we have $\mathds P(\Omega_3(\eta,c_2)) > 1 - 6\eta/8$. Furthermore, 
	\begin{align*}
	\rho & = \int\limits_{(0, \widetilde E_N]} \mathcal{B}(E - \mu_N^{\omega}) \, \mathcal N_N^{\omega}(\mathrm{d} E) + \int\limits_{(\widetilde E_N,[2 \pi \nu / \ln(L_N)]^2]} \mathcal{B}(E - \mu_N^{\omega}) \, \mathcal N_N^{\omega}(\mathrm{d} E) \\
	& \quad + \, \int\limits_{([2 \pi \nu / \ln(L_N)]^2,\epsilon]} \mathcal{B}(E - \mu_N^{\omega}) \, \mathcal N_N^{\omega}(\mathrm{d} E) + \int\limits_{(\epsilon,\infty)} \mathcal{B}(E - \mu_N^{\omega}) \, \mathcal N_N^{\omega}(\mathrm{d} E)
	\end{align*}
	for all $\epsilon > 0$ and all $N \in \mathds N$ large enough. Consequently, 
	\begin{align}
	& \E \left[ \int\limits_{(0,\widetilde E_N]} \mathcal{B}(E - \mu_N^{\omega}) \, \mathcal N_N^{\omega}(\mathrm{d} E) \right] = \\
	& \quad = \rho - \int\limits_{\Omega_3(\eta,c_2)} \left[ \, \int\limits_{(\widetilde E_N,[2 \pi \nu / \ln(L_N)]^2]} \mathcal{B}(E - \mu_N^{\omega})\, \mathcal N_N^{\omega}(\mathrm{d} E) \right] \, \mathds P(\mathrm{d} \omega) \label{proof BEC type ns LSM 1}\\
	& \qquad \ \ \  - \int\limits_{\Omega_3(\eta,c_2)} \left[ \, \int\limits_{([2 \pi \nu / \ln(L_N)]^2,\epsilon]} \mathcal{B}(E - \mu_N^{\omega}) \, \mathcal N_N^{\omega}(\mathrm{d} E) \right] \, \mathds P(\mathrm{d} \omega) \label{proof BEC type ns LSM 2} \\
	& \qquad \ \ \ - \int\limits_{\Omega \backslash \Omega_3(\eta,c_2)} \left[ \, \int\limits_{(\widetilde E_N,\epsilon]} \mathcal{B}(E - \mu_N^{\omega}) \, \mathcal N_N^{\omega}(\mathrm{d} E) \right] \, \mathds P(\mathrm{d} \omega) \label{proof BEC type ns LSM 3} \\
	& \qquad \ \ \ - \E \left[ \, \int\limits_{(\epsilon,\infty)} \mathcal{B}(E - \mu_N^{\omega}) \, \mathcal N_N^{\omega}(\mathrm{d} E) \right] \label{proof BEC type ns LSM 4}
	\end{align}
	for all $\epsilon > 0$ and all $N \in \mathds N$ large enough.
	
Now, in a first step we realize that 
	\begin{align*}
	\lim\limits_{\epsilon \searrow 0} \limsup\limits_{N \to \infty} \mathds E \left[ \, \int\limits_{(\epsilon,\infty)} \mathcal{B}(E - \mu_N^{\omega}) \mathcal N_N^{\omega} (\mathrm{d} E) \right] \le \rho_c(\beta) \ .
	\end{align*}
	This can be seen as follows: for $\epsilon > 0$ we have $\int_{(\epsilon,\infty)} \mathcal{B}(E - \mu_N^{\omega})\mathcal N_N^{\omega} (\mathrm{d} E) \leq \rho$. Moreover, 
	\begin{align*}
	\limsup\limits_{N \to \infty} \int\limits_{(\epsilon,\infty)} \mathcal{B}(E - \mu_N^{\omega}) \mathcal N_N^{\omega} (\mathrm{d} E) \le \int\limits_{(\epsilon,\infty)} \mathcal{B}(E) \mathcal N_{\infty} (\mathrm{d} E) + \dfrac{2}{\beta \epsilon} \mathcal N_{\infty}^{\mathrm{I}}(\epsilon)
	\end{align*}
	$\mathds{P}$-almost surely due to~\eqref{LemmaA6limsup}. Hence, by (reverse) Fatou's Lemma
	\begin{align*}
	 \limsup\limits_{N \to \infty} \mathds E \left[  \int\limits_{(\epsilon,\infty)} \mathcal{B}(E - \mu_N^{\omega}) \mathcal N_N^{\omega} (\mathrm{d} E) \right] \le \int\limits_{(\epsilon,\infty)} \mathcal{B}(E) \mathcal N_{\infty} (\mathrm{d} E) + \dfrac{2}{\beta \epsilon} \mathcal N_{\infty}^{\mathrm{I}}(\epsilon) \ .
	\end{align*}
We conclude
	\begin{align*}
	\lim\limits_{\epsilon \searrow 0} \limsup\limits_{N \to \infty} \mathds E \left[ \int\limits_{(\epsilon,\infty)} \mathcal{B}(E - \mu_N^{\omega}) \mathcal N_N^{\omega} (\mathrm{d} E) \right] 	& \le \lim\limits_{\epsilon \searrow 0} \int\limits_{(\epsilon,\infty)} \mathcal{B}(E) \mathcal N_{\infty} (\mathrm{d} E)  \\
	& = \int\limits_{(0,\infty)} \mathcal{B}(E) \mathcal N_{\infty} (\mathrm{d} E) \\
	& = \rho_c(\beta)
	\end{align*}
	due to Theorem~\ref{Lifshitz Auslaufer one dimensional} and Lemma~\ref{finite critical density in LSmodel}.

	In a next step we show that the integrals in lines \eqref{proof BEC type ns LSM 1}--\eqref{proof BEC type ns LSM 3} either vanish or can be bounded from above in the considered limit: concerning~\eqref{proof BEC type ns LSM 1} we observe that for all but finitely many $N \in \mathds N$ and for all $\omega \in \Omega_3(\eta,c_2)$ we obtain 
	\begin{align} \label{Hinweis im beweis von theorem 3.4}
	\widetilde E_N  - E_N^{1,\omega} \ge \pi^2 \nu^2 \dfrac{2\ln(c_1 L_N) \ln(c_2/2) - [\ln(c_2/2)]^2}{[\ln(c_1 L_N)]^4} \ge \dfrac{\pi^2 \nu^2}{[\ln(c_1 L_N)]^3} \cdot \ln(c_2/2)
	\end{align}
	and therefore
	\begin{align*}
	\mathcal{B}(E - \mu_N^{\omega}) \le \left( \beta (E - E_N^{1,\omega}) \right)^{-1} \le \dfrac{[\ln(c_1 L_N)]^3}{\beta \pi^2 \nu^2} \cdot \dfrac{1}{\ln(c_2/2)}
	\end{align*}
	for all $E \ge \widetilde E_N$. Hence, 
	\begin{align*}
	A & :=\lim\limits_{N \to \infty} \int\limits_{\Omega_3(\eta,c_2)} \left[ \int\limits_{(\widetilde E_N,[2 \pi \nu/ \ln (L_N)]^2]} \mathcal{B}(E - \mu_N^{\omega}) \, \mathcal N_N^{\omega} (\mathrm{d} E) \right] \mathds P(\mathrm{d} \omega) \\
	& \le \lim\limits_{N \to \infty} \int\limits_{\Omega_3(\eta,c_2)} \left[ \int\limits_{(\widetilde E_N,[2 \pi \nu/ \ln (L_N)]^2]} \dfrac{[\ln(c_1 L_N)]^3}{\beta \pi^2 \nu^2} \cdot \dfrac{1}{\ln(c_2/2)} \, \mathcal N_N^{\omega} (\mathrm{d} E) \right] \mathds P(\mathrm{d} \omega) \\
	& \le \lim\limits_{N \to \infty} \dfrac{[\ln(c_1 L_N)]^3}{\beta \pi^2 \nu^2} \cdot \dfrac{1}{\ln(c_2/2)} \int\limits_{\Omega_3(\eta,c_2)} \mathcal N_N^{\mathrm{I},\omega} \left( \dfrac{4 \pi^2 \nu^2}{ [\ln (L_N)]^{2}} \right) \mathds P(\mathrm{d} \omega) \ ,
	\end{align*}
and since $\mathcal N_N^{\mathrm{I},\omega}(E) \ge 0$ for all $E > 0$ and all $\omega \in \Omega_3(\eta,c_2)$ we continue, employing Lemma~\ref{lemma E2} ($\mathcal E = (\pi \nu )^{-1} \ln(M^{1/2})$) and Theorem~\ref{Lifshitz Auslaufer one dimensional} ($\widetilde{M}=M^{1/2}$),
\begin{align*}
  A \le \, & \lim\limits_{N \to \infty} \dfrac{[\ln(c_1 L_N)]^3}{\beta \pi^2 \nu^2} \cdot \dfrac{1}{\ln(c_2/2)} \E \left[ \mathcal N_N^{\mathrm{I},\omega} \left( \dfrac{4 \pi^2 \nu^2}{ [\ln (L_N)]^{2}} \right)  \right] \\
  \le \, & \lim\limits_{N \to \infty} \dfrac{[\ln(c_1 L_N)]^3}{\beta \pi^2 \nu^2} \cdot \dfrac{1}{\ln(c_2/2)} \dfrac{M}{L_N^{1/2}}  = 0 \ .
\end{align*}
We now turn to line~\eqref{proof BEC type ns LSM 2}: We firstly observe that for all $E \ge [2 \pi \nu / \ln(L_N)]^2$ and $\omega \in \Omega_3(\eta,c_2)$ we have $E \ge [2 \pi \nu / \ln(L_N)]^2 \ge 2 E_N^{1,\omega}$ and hence $E - \mu_N^{\omega} \ge E - E_N^{1,\omega} \ge (1/2) E$ for all but finitely many $N \in \mathds N$. Therefore,
\begin{align*}
B& :=\lim\limits_{\epsilon \searrow 0} \lim\limits_{N \to \infty} \int\limits_{\Omega_3(\eta,c_2)} \left[ \int\limits_{([2 \pi \nu/ \ln (L_N)]^2,\epsilon]}\mathcal{B}(E - \mu_N^{\omega}) \mathcal N_N^{\omega} (\mathrm{d} E) \right] \mathds P(\mathrm{d} \omega) \\
& \le\lim\limits_{\epsilon \searrow 0} \lim\limits_{N \to \infty} \int\limits_{\Omega_3(\eta,c_2)} \left[ \int\limits_{([2 \pi \nu/ \ln (L_N)]^2,\epsilon]} \mathcal B(E/2)\, \mathcal N_N^{\omega} (\mathrm{d} E) \right] \mathds P(\mathrm{d} \omega) \\
& \le \dfrac{2}{\beta}  \lim\limits_{\epsilon \searrow 0} \lim\limits_{N \to \infty} \int\limits_{\Omega_3(\eta,c_2)} \left[ \int\limits_{([2 \pi \nu/ \ln (L_N)]^2,\epsilon]} E^{-1}  \mathcal N_N^{\omega} (\mathrm{d} E) \right] \mathds P(\mathrm{d} \omega) \ .
\end{align*}
Then, an integration by parts (for Lebesgue--Stieltjes integrals, see, e.g., \cite[Theorem 21.67]{hewitt1965real}) and an application of the Fubini--Tonelli theorem yields
\begin{align*}
B & \le \dfrac{2}{\beta}  \lim\limits_{\epsilon \searrow 0} \lim\limits_{N \to \infty} \int\limits_{\Omega_3(\eta,c_2)} \left[ \epsilon^{-1} \mathcal N_N^{\mathrm{I},\omega} (\epsilon) + \int\limits_{[2 \pi \nu / \ln(L_N)]^2}^{\epsilon} \mathcal N_N^{\mathrm{I},\omega} (E) E^{-2} \, \mathrm{d} E \right] \, \mathds P(\mathrm{d} \omega) \\ 
& \le \dfrac{2}{\beta}  \lim\limits_{\epsilon \searrow 0} \lim\limits_{N \to \infty} \Bigg[ \epsilon^{-1} \int\limits_{\Omega_3(\eta,c_2)} \mathcal N_N^{\mathrm{I},\omega} ( \epsilon) \, \mathds P(\mathrm{d} \omega) \, + \\
& \qquad \qquad \qquad \qquad + \, \left. \int\limits_{[2 \pi \nu / \ln(L_N)]^2}^{\epsilon} \left( \, \int\limits_{\Omega_3(\eta,c_2)} \mathcal N_N^{\mathrm{I},\omega} (E) \, \mathds P(\mathrm{d} \omega)  \right) \, E^{-2} \, \mathrm{d} E \right] \\
& \le \dfrac{2}{\beta}  \lim\limits_{\epsilon \searrow 0} \lim\limits_{N \to \infty} \left[ \epsilon^{-1} \E \left[ \mathcal N_N^{\mathrm{I},\omega} ( \epsilon) \right] + \int\limits_{[2 \pi \nu / \ln(L_N)]^2}^{\epsilon} \E \left[ \mathcal N_N^{\mathrm{I},\omega} (E) \right]  E^{-2} \, \mathrm{d} E \right]\ .
\shortintertext{With an $M > 0$ due to Lemma~\ref{lemma E2} $(\mathcal{E}=(\pi\nu)^{-1}\ln(M^{1/2}))$ and Theorem~\ref{Lifshitz Auslaufer one dimensional} ($\widetilde M = M^{1/2}$)}
& \le \dfrac{2 M}{\beta}  \lim\limits_{\epsilon \searrow 0} \lim\limits_{N \to \infty} \left[ \epsilon^{-1} \mathrm{e}^{-\nu \pi \epsilon^{-1/2}} + \int\limits_{[2 \pi \nu / \ln(L_N)]^2}^{\epsilon} \mathrm{e}^{-\nu \pi E^{-1/2}}  E^{-2} \, \mathrm{d} E \right] \\
& \le \dfrac{2 M}{\beta}  \lim\limits_{\epsilon \searrow 0} \lim\limits_{N \to \infty} \left[ \epsilon^{-1} \dfrac{4! \, \epsilon^{2} }{ (\nu \pi)^4 } + \int\limits_{[2 \pi \nu / \ln(L_N)]^2}^{\epsilon} \dfrac{4! \, E^2}{(\nu \pi)^4} E^{-2} \, \mathrm{d} E \right]  \\
& \le \dfrac{2 M}{\beta} \dfrac{4! }{ (\nu \pi)^4 } \lim\limits_{\epsilon \searrow 0} \left[ \epsilon + \epsilon \right] = 0
\end{align*}
	$\mathds{P}$-almost surely. We now turn to line~\eqref{proof BEC type ns LSM 3}: we calculate, using (reverse) Fatou's Lemma, dominated convergence, and Lemma~\ref{Lemma beweis lim int epsilon infty und lim lim int epsilion infty rhoc vage Konvergenz},
	\begin{align*}
	& \lim\limits_{\epsilon \searrow 0} \limsup\limits_{N \to \infty} \int\limits_{\Omega \backslash \Omega_3(\eta,c_2)} \left[ \int_{(0,\epsilon]} \mathcal{B}(E - \mu_N^{\omega})\, \mathcal N_N^{\omega} (\mathrm{d} E) \right] \mathds P(\mathrm{d} \omega) \\
& \quad \le \lim\limits_{\epsilon \searrow 0} \int\limits_{\Omega \backslash \Omega_3(\eta,c_2)} \limsup\limits_{N \to \infty} \left[ \int_{(0,\epsilon]}\mathcal{B}(E - \mu_N^{\omega}) \, \mathcal N_N^{\omega} (\mathrm{d} E) \right] \mathds P(\mathrm{d} \omega) \\
	& \quad = \int\limits_{\Omega \backslash \Omega_3(\eta,c_2)} \lim\limits_{\epsilon \searrow 0} \limsup\limits_{N \to \infty} \left[ \int_{(0,\epsilon]} \mathcal{B}(E - \mu_N^{\omega}) \, \mathcal N_N^{\omega} (\mathrm{d} E) \right] \mathds P(\mathrm{d} \omega) \leq \rho_0(\beta) \eta \ .
	\end{align*}
Altogether we have shown that 
	\begin{align} \label{lim E XN ge 1 eta rho0}
	& \liminf\limits_{N \to \infty} \E \left[ \int\limits_{(0,\widetilde E_N]} \mathcal{B}(E - \mu_N^{\omega}) \, \mathcal N_N^{\omega}(\mathrm{d} E) \right] \ge (1 - \eta) \rho_0(\beta).
	\end{align}
In order to proceed, we define the random variable 
	\begin{align*}
	X_N^{\omega} := (1 - \eta) \rho_0(\beta) - \int\limits_{(0,\widetilde E_N]} \mathcal{B}(E - \mu_N^{\omega})\, \mathcal N_N^{\omega} (\mathrm{d} E) \ .
	\end{align*}
	By Lemma~\ref{Lemma beweis lim int epsilon infty und lim lim int epsilion infty rhoc vage Konvergenz} and Theorem~\ref{Lifshitz Auslaufer one dimensional} we obtain, for some fixed $\epsilon' > 0$ small enough,
	\begin{align*}
	\limsup_{N \to \infty} \int\limits_{(0,\widetilde E_N]} \mathcal{B}(E - \mu_N^{\omega}) \, \mathcal N_N^{\omega} (\mathrm{d} E) & \le \limsup\limits_{N \to \infty} \int\limits_{(0,\epsilon']} \mathcal{B}(E - \mu_N^{\omega}) \, \mathcal N_N^{\omega} (\mathrm{d} E) \\
	& \le \left(1 + \dfrac{\eta}{2} \right) \rho_0(\beta) \ ,
	\end{align*}
	and hence, with $[ x ]_+:=\max\left(x,0\right)$,
    \begin{align*}	 
     \lim\limits_{N \to \infty} \left[ \int\limits_{(0,\widetilde E_N]} \mathcal{B}(E - \mu_N^{\omega}) \, \mathcal N_N^{\omega} (\mathrm{d} E) - \left( 1 + \dfrac{\eta}{2} \right) \rho_0(\beta) \right]_+ = 0
    \end{align*}
	$\mathds{P}$-almost surely. Consequently, 
	\begin{align*}
	\mathds P \left( \omega: \left[ \int\limits_{(0,\widetilde E_N]} \mathcal{B}(E - \mu_N^{\omega}) \, \mathcal N_N^{\omega} (\mathrm{d} E) - \left(1 + \dfrac{\eta}{2} \right) \rho_0(\beta) \right]_+ \le \dfrac{\eta}{2}\rho_0(\beta) \right) \ge 1 - \eta
	\end{align*}
	and therefore
	$\mathds P (\omega: X_N^{\omega} \ge - 2 \eta \rho_0(\beta)) \ge 1 - \eta$ for all but finitely many $N \in \mathds N$. 

	Moreover, $X_N^{\omega} \ge - \rho$ for all $N \in \mathds N$ and $\mathds{P}$-almost all $\omega \in \Omega_3(\eta,c_2)$. Also,~\eqref{lim E XN ge 1 eta rho0} implies $\limsup_{N \to \infty} \E[X_N^{\omega}] \le 0$ and therefore, for all but finitely many $N \in \mathbb{N}$,
	\begin{align*}\begin{split}
	 \eta & \geq \E \left[ X_N^{\omega} \right] 
	=\int\limits_{X_N^{\omega} \ge \eta^{1/2} (1 - \eta) \rho_0(\beta)} X_N^{\omega} \, \mathds P(\mathrm{d} \omega) + \int\limits_{-2 \eta \rho_0(\beta) \le X_N^{\omega} < \eta^{1/2} (1 - \eta) \rho_0(\beta)} X_N^{\omega} \, \mathds P(\mathrm{d} \omega) \\ & \qquad \qquad \qquad \qquad + \int\limits_{X_N^{\omega} < -2 \eta \rho_0(\beta)} X_N^{\omega} \, \mathds P(\mathrm{d} \omega) \\
	 & \geq \eta^{1/2} (1 - \eta) \rho_0(\beta) \, \mathds P\Big( X_N^{\omega} \ge \eta^{1/2} (1 - \eta) \rho_0(\beta) \Big) - 2 \eta \rho_0(\beta) - \eta \rho\ . 
	 \end{split}
	\end{align*}
	Since $0 < \eta \le 1/2$,
	\begin{align*}
	\mathds P \left( \omega: \int\limits_{(0,\widetilde E_N]} \mathcal{B}(E - \mu_N^{\omega}) \, \mathcal N_N^{\omega}(\mathrm{d} E) \le (1 - \eta^{1/2}) (1 - \eta) \rho_0(\beta) \right) \le 4 \dfrac{\rho_0(\beta) + \rho + 1}{\rho_0(\beta)} \eta^{1/2}
	\end{align*}
	for all but finitely many $N \in \mathds N$. 
	
	Hence, in total we have shown that for all $0 < \eta \le 1/2$ there exists an $\widetilde N = \widetilde N(\eta) \in \mathds N$ such that for all $N \ge \widetilde N$ one has 
	\begin{align*}
	& \mathds P \left( \omega: \dfrac{n_N^{1,\omega}}{N}  \ge \dfrac{1}{c_3} \rho^{-1} (1 - \eta^{1/2}) (1 - \eta) \rho_0(\beta) \right) \\
	& \quad \ge \mathds P \left( \left\{ \omega: \dfrac{n_N^{1,\omega}}{N}  \ge \dfrac{1}{c_3} \rho^{-1} (1 - \eta^{1/2}) (1 - \eta) \rho_0(\beta) \right\} \cap \Omega_3(\eta,c_2) \right)\\
	& \quad \ge \mathds P \Bigg( \Big\{ \omega: \int\limits_{(0,\widetilde E_N]} \mathcal{B}(E - \mu_N^{\omega}) \, \mathcal N_N^{\omega}(\mathrm{d} E) \ge  (1 - \eta^{1/2}) (1 - \eta) \rho_0(\beta) \Big\} \cap \Omega_3(\eta,c_2) \Bigg) \\ 
	&\ge 1 - 4 \dfrac{\rho_0(\beta) + \rho + 1}{\rho_0(\beta)} \eta^{1/2} - 6\eta/8 \ .
	\end{align*}
	Note that we used the fact that $n_N^{j,\omega}/N$ is monotonically decreasing in $j$. 
	\end{proof}
Theorem~\ref{LSM ns main theorem} shows that the one-particle ground state is, with strictly positive probability, macroscopically occupied. Regarding excited states we can show the following.
	\begin{cor}\label{CorollaryMinimal}
        Let $M>0$ be given as in Theorem \ref{LSM ns main theorem} and $c_1$ as in Theorem \ref{satz E2}. Then,
	 for all $\eta' > 0$, $0 < \eta \le 1/2$, $c_2 > 2$, and $j \ge c_3:=\lceil 4 M c_2 / (\eta c_1) \rceil + 1$ there exists an $\widetilde N = \widetilde N(\eta', \eta, c_2)$ such that for all $N \ge \widetilde N$ one has
	 \begin{align*}
	  \mathds P \left( \omega: \dfrac{n_N^{j,\omega}}{N} < \eta' \right) \ge \mathds P \left( \omega: \dfrac{n_N^{c_3,\omega}}{N} < \eta' \right) > 1-5\eta/8\,.
	 \end{align*}
	\end{cor}
	\begin{proof}
	 Let $\eta' > 0$, $0<\eta<2$ and $c_2>2$ be arbitrary. According to Theorem~\ref{satz E2} there exists an $\widetilde N = \widetilde N(\eta,c_2) \in \mathds N$ such that for all $N \ge \widetilde N$ we have $\mathds P(\Omega_2(\eta,c_2)) > 1 - 5 \eta/8$
	 with $\Omega_2(\eta,c_2) $ as in Theorem~\ref{satz E2}.
	 	
	 Consequently, there exists a number $\widehat{N} \geq \widetilde N$ such that, for all $N \ge \widehat{N}$ and all $\omega \in \Omega_2(\eta,c_2)$ (see also \eqref{Hinweis im beweis von theorem 3.4}) 
  \begin{align*}
   \dfrac{1}{N} n_N^{c_3,\omega} & = \dfrac{1}{N} \left( \mathrm{e}^{\beta (E_N^{c_3,\omega} - \mu_N^{\omega})} - 1 \right)^{-1} \\
   & \le \dfrac{1}{N} \left( \mathrm{e}^{\beta (E_N^{c_3,\omega} - E_N^{1,\omega})} - 1 \right)^{-1} \\
   & \le \dfrac{1}{N} \Big[ \beta (E_N^{c_3,\omega} - E_N^{1,\omega}) \Big]^{-1} 
   < \eta' \ .
  \end{align*}
	\end{proof}
	Most importantly, Theorem~\ref{LSM ns main theorem} implies $\mathds{P}$-almost sure Bose--Einstein condensation into the one-particle ground state as demonstrated by the following statement. 
	\begin{theorem}[Macroscopic occupation of the ground state]\label{MacroscopicOccupationGroundstate}
		The ground state is $\mathds{P}$-almost surely macroscopically occupied, that is, according to Definition \ref{Definition makroskopische Besetzung}, 
		\begin{align*}
		\mathds P \left( \omega: \limsup\limits_{N \to \infty} \dfrac{n_N^{1,\omega}}{N}  > 0 \right) = 1 \ .
		\end{align*}
	\end{theorem}
	\begin{proof}
		Suppose to the contrary that there exists a $0 < C \le 1$ such that $\mathds P(\lim_{N \to \infty} n_N^{1,\omega} / N = 0) \ge C$. 
		Then, for $0 < \eta \leq 1/2 $, $c_2 > 2$, and $c_3 = c_3(\eta,M)$ as in Theorem~\ref{LSM ns main theorem} we obtain
		\begin{align*}\begin{split}
		& \limsup\limits_{N \to \infty} \mathds{P}\left(\omega:  \dfrac{n_N^{1,\omega}}{N} \ge \dfrac{1}{c_3} \rho^{-1}(1 - \eta^{1/2}) (1 - \eta) \rho_0(\beta)\right) \\
		& \quad \le \mathds P \left( \omega: \limsup\limits_{N \to \infty} \dfrac{n_N^{1,\omega}}{N} \ge \dfrac{1}{c_3} \rho^{-1}(1 - \eta^{1/2}) (1 - \eta) \rho_0(\beta) \right) \\
		& \quad \le \mathds P \left( \omega: \neg \left( \lim\limits_{N \to \infty} \dfrac{n_N^{1,\omega}}{N} = 0 \right) \right) \le 1 - C \ .
		\end{split}
		\end{align*}
		However, due to Theorem~\ref{LSM ns main theorem} there exists an $\widetilde N \in \mathds N$ such that for all $N \ge \widetilde N$ one has
		\begin{align*}
	& \mathds P \left( \omega: \dfrac{n_N^{1,\omega}}{N}  \ge \dfrac{1}{c_3} \rho^{-1}(1 - \eta^{1/2})(1 - \eta) \rho_0(\beta) \right) \\
		& \quad \ge 1 - 4 \dfrac{\rho_0(\beta)+ \rho + 1}{\rho_0(\beta)} \eta^{1/2} - 6\eta/8 \\
		& \quad \ge 1 - \dfrac{C}{2} \ ,
		\end{align*}
		for $\eta$ small enough. 
	\end{proof}

\begin{remark} \label{remark 3.6} Using Theorem \ref{LSM ns main theorem} we can show that
\begin{align*} \mathds E\Big[\dfrac{n_N^{1,\omega}}{N}\Big]& = \mathds E\Big[\dfrac{n_N^{1,\omega}}{N} \mathds{1}\Big(\dfrac{n_N^{1,\omega}}{N} \ge c(\eta)\Big)\Big] + \mathds E\Big[\dfrac{n_N^{1,\omega}}{N} \mathds{1}\Big(\dfrac{n_N^{1,\omega}}{N} < c(\eta)\Big)\Big]\ge c(\eta)(1-\epsilon)\,,
\end{align*}
where $c(\eta) = \dfrac{1}{c_3} \rho^{-1}(1 - \eta^{1/2})(1 - \eta) \rho_0(\beta)$, $\epsilon = 4 \dfrac{\rho_0(\beta) + \rho + 1}{\rho_0(\beta)} \eta^{1/2} + 6\eta/8$, and $\mathds{1}(\cdot)$ the indicator function. This proves that $\liminf_{N\to\infty} \mathds E\big[{n_N^{1,\omega}}/{N}\big] >0$, that is, there is macroscopic occupation of the ground state in expectation.
\end{remark}

\vspace*{0.5cm}

\subsection*{Conclusion}{In this paper we proved, for the first time, that the one-particle ground state in the Luttinger--Sy model with non-infinite interaction strength is macroscopically occupied, see Theorem~\ref{MacroscopicOccupationGroundstate}. This follows from Theorem~\ref{LSM ns main theorem}, where we show that the probability of macroscopic occupation of the ground state is arbitrarily close to $1$, uniformly in the number of particles. It seems plausible but we are not able to prove that only the ground state is macroscopically occupied; this would then yield a type-I BEC. Nevertheless, in Corollary~\ref{CorollaryMinimal} we do prove that highly excited one-particle eigenstates are arbitrarily close to not being occupied with a probability arbitrarily close to $1$. 

\subsection*{Acknowledgement}{It is our pleasure to thank Werner Kirsch and Hajo Leschke for interesting discussions and useful remarks that led to an improvement of the manuscript.}

\appendix
%
%
%
%
\section{Appendix}\label{secAppendix}
In this appendix we collect results for the LS model that we used to establish the results in the previous sections. We start with an estimate between the integrated density of states for finite $N$ and the infinite-volume integrated density of states. 
\begin{lemma}\label{lemma E2}For all $N \in \mathds N$, $E > 0$, and $0 < \mathcal E < E^{-1/2}$ we have
	\begin{align}
	\mathds E \left[ \mathcal N_N^{\mathrm{I},\omega}(E) \right] \le \mathcal N_{\infty}^{\mathrm{I}} ([E^{-1/2} - \mathcal E]^{-2}) \ .
	\end{align}
Moreover, for all $N \in \mathds N$ and Lebesgue-almost all $E \ge 0$ (i.e., all $E \ge 0$ except the discontinuity points of $\mathcal N_{\infty}^{\mathrm{I}}$) one has
	\begin{align}\label{EquationAppendixFirstTheoremSecond}
	\mathds E \left[ \mathcal N_N^{\mathrm{I},\omega}(E) \right] \le \mathcal N_{\infty}^{\mathrm{I}} (E) \ .
	\end{align}
\end{lemma}
\begin{proof} In this proof, $\mathcal N_{\Lambda}^{\mathrm{I},\omega}$ denotes the integrated density of states corresponding to the operator $h^{\Lambda}_{\gamma}(\omega)$ on $H^1_0(\Lambda)$, defined analogously to~\eqref{FormalOperatorLSModel}.

In a first step, we realize that, for arbitrary disjoint intervals $\Lambda^1, \Lambda^2 \subset \mathds R$ and for any $E \ge 0$, the inequality 
	\begin{align} \label{akjzsdjk}
	|\Lambda^1|\, \mathcal N_{\Lambda^1}^{\mathrm{I}, \omega}(E) + |\Lambda^2| \, \mathcal N_{\Lambda^2}^{\mathrm{I}, \omega}(E) \le |\Lambda^1 \cup \Lambda^2| \, \mathcal N_{\Lambda^1 \cup \Lambda^2}^{\mathrm{I}, \omega}(E)
	\end{align}
	holds for $\mathds{P}$-almost all $\omega \in \Omega$, see, e.g., \cite[5.39a]{pastur1992spectra}.

	For $M,N \in \mathds{N}$ and $j \in \mathds Z$ we define $\Lambda^{j}_N:=\Lambda_N + L_N \cdot j$ and $\Lambda_{M,N}:=\bigcup_{j \in \mathds Z: |j| \le M} \Lambda_N^j$.
	Note that $\{ \mathcal N_{\Lambda_N^j}^{\mathrm{I},\omega}(E) \}_{j \in \mathds Z}$ is a set of independent and identically distributed random variables for all $E \ge 0$. Hence, employing the strong law of large numbers, inequality \eqref{akjzsdjk}, and introducing the continuous function
	\begin{align*}
	f_{E, \mathcal E} : \mathds R \to \mathds R,\ x \mapsto \begin{cases}
	0 & \text{ if } x \le - \mathcal E \\
	1 + \mathcal E^{-1} x & \text{ if } -\mathcal E < x < 0 \\
	1 & \text{ if } 0 \le x \le E \\
	1 - \dfrac{x-E}{[E^{-1/2} - \mathcal E]^{-2} - E} & \text{ if } E < x < [E^{-1/2} - \mathcal E]^{-2} \\
	0 & \text{ if } x \ge [E^{-1/2} - \mathcal E]^{-2}
	\end{cases}
	\end{align*}
	we find that, for all $E > 0$,
	\begin{align*}
	\mathds E \Big[ \mathcal N_N^{\mathrm{I},\omega}(E) \Big] & = \lim\limits_{M \to \infty} \dfrac{1}{2M+1}\sum\limits_{j \in \mathds Z : |j| \le M} \mathcal N_{\Lambda_N^j}^{\mathrm{I},\omega}(E)\\
	& = \lim\limits_{M \to \infty} \dfrac{1}{|\Lambda_M|} \sum\limits_{j \in \mathds Z : |j| \le M} |\Lambda_N^j| \, \mathcal N_{\Lambda_N^j}^{\mathrm{I},\omega}(E) \\
	& \le \limsup\limits_{M \to \infty} \mathcal N_{\Lambda_{M,N}}^{\mathrm{I},\omega}(E) \le \limsup\limits_{M \to \infty} \mathcal N_{M}^{\mathrm{I},\omega}(E) \\
	& \le \lim\limits_{M \to \infty} \int\limits_{\mathds R} f_{E,\mathcal E}(\widetilde E) \, \mathcal N_M^{\omega}(\mathrm{d} \widetilde E) = \int\limits_{\mathds R} f_{E,\mathcal E}(\widetilde E) \, \mathcal N_{\infty}(\mathrm{d} \widetilde E) \\
	& \le \mathcal N_{\infty}^{\mathrm{I}}([E^{-1/2} - \mathcal E]^{-2}) \ .
	\end{align*}	
%
\eqref{EquationAppendixFirstTheoremSecond} follows from the fact that $\mathcal  N_{\infty}^{\mathrm{I}}$ is monotonically increasing and hence has at most countably many points of discontinuity and by taking the limit $\mathcal{E} \searrow 0$.
\end{proof}
Using this result we obtain a lower bound for the ground-state energy.
  
  \begin{lemma} \label{lower bound ground state energy} For all $\kappa > 2$ and for $\mathds{P}$-almost all $\omega$ there exists an $\widetilde N = \widetilde N(\kappa, \omega) \in \mathds N$ such that for all $N \ge \widetilde N$ we have
  	\begin{align}
  	E_N^{1,\omega} \ge \left( \dfrac{\pi\nu }{\kappa \ln(L_N)} \right)^{2} \ .
  	\end{align}
  \end{lemma}
  \begin{proof}
  	Let $\kappa > 2$ be given. We define $\widehat E_N := (\pi\nu  / [ \kappa \ln(L_N)])^{2}$ for all $N \in \mathds N$ with $L_N> 1$, and pick some $\mathcal E > 0$.
  	
  	Then, with Lemma~\ref{lemma E2} and Theorem~\ref{Lifshitz Auslaufer one dimensional} we conclude that for all but finitely many $N \in \mathds N$ one has
  	\begin{align*}
  	\mathds P\left(\omega: |\Lambda_N| \cdot \mathcal N_N^{\omega}(\widehat E_N) \ge 1 \right) & \le L_N \cdot \E \left[ \mathcal N_N^{\mathrm{I}, \omega}(\widehat E_N) \right] \\
  	& \le L_N \cdot \mathcal N^{\mathrm{I}}_{\infty}([\widehat E_N^{-1/2} - \mathcal E]^{-2}) \\
  	& \le \widetilde{M}L_N  \cdot \exp\left( - \pi\nu \left[ \widehat E_N^{-1/2} - \mathcal E \right]  \right) \\
  	& \le \widetilde{M} \mathrm{e}^{\pi \nu  \mathcal E}\cdot L_N^{- \kappa+1} \ .
  	\end{align*}
  	Hence $\sum_{N=1}^{\infty} \mathds P\left(\omega:  L_N\cdot \mathcal N_N^{\omega}(\widehat E_N) \ge 1 \right) < \infty$ and the statement follows from the Borel--Cantelli lemma.
  \end{proof}

\begin{remark}
 Lemma~\ref{lemma E2} and Lemma~\ref{lower bound ground state energy} can readily be generalized to Poisson random potentials of the form
 $$V(\omega, \cdot) = \sum\limits_{j\in\mathds{Z}} u(\cdot - x_j(\omega)) \ ,$$
 with $u \in L^{\infty}(\mathds R)$ having compact support. What we do not have at our disposal is the $O(E^{1/2})$-error bound in the Lifshitz-tail behavior. If we did we could carry over our results to these models.
\end{remark}

Note that the critical density $\rho_c(\beta)$ was defined in~\eqref{EquationCriticalDensity}.

\begin{lemma} \label{finite critical density in LSmodel}
For all $\beta > 0$, the critical density $\rho_c(\beta)$ satisfies
 \begin{align*}
  \rho_c(\beta) = \int\limits_{(0,\infty)} \mathcal{B}(E) \, \mathcal N_{\infty} (\mathrm{d} E) < \infty \ .
 \end{align*}
\end{lemma}
 \begin{proof}
In a first step we show that
\begin{align}\label{EquationProofRelationXXXXX}
 \int\limits_{\mathds R} \mathcal{B}(E) \, \mathcal N_{\infty} (\mathrm{d} E) = \int\limits_{(0,\infty)} \mathcal{B}(E) \, \mathcal N_{\infty} (\mathrm{d} E)
\end{align}
%
is finite. In order to do this, we choose $\widetilde E_1: = \widetilde E$ with  $\widetilde E > 0$ as in Theorem~\ref{Lifshitz Auslaufer one dimensional}. Moreover, we choose $\widetilde E_2 > \widetilde E_1$ such that for $E \ge \widetilde E_2$ one has
\begin{align} \label{Ungleichung Appendix 2 kritische Dichte Energieschranke 1}
\mathcal{B}(E) \le \left( 2^{-1/2} \mathrm{e}^{\beta E}  \right)^{-1}
\end{align}
and
\begin{align} \label{Ungleichung Appendix 2 kritische Dichte Energieschranke 2}
 \mathcal N_{\infty}^{\mathrm{I}}(E) \le c E^{1/2} 
\end{align}
with $c > 0$ as in Theorem~\ref{Lifshitz Auslaufer one dimensional}. We write,
\begin{align} 
 \int\limits_{(0,\infty)} \mathcal{B}(E) \, \mathcal N_{\infty} (\mathrm{d} E) & = \int\limits_{(0,\widetilde E_1]} \mathcal{B}(E) \, \mathcal N_{\infty} (\mathrm{d} E)  \label{Appendix proof 2 critical density align 1 1} \\
 & + \int\limits_{(\widetilde E_1,\widetilde E_2]}\mathcal{B}(E) \, \mathcal N_{\infty} (\mathrm{d} E) \label{Appendix proof 2 critical density align 1 2} \\ 
 & + \int\limits_{(\widetilde E_2,\infty)} \mathcal{B}(E) \, \mathcal N_{\infty} (\mathrm{d} E) \label{Appendix proof 2 critical density align 1 3}
 \end{align}
  The term in line~\eqref{Appendix proof 2 critical density align 1 2} is finite since
 \begin{align*}
   \int\limits_{(\widetilde E_1,\widetilde E_2]} \mathcal{B}(E) \, \mathcal N_{\infty} (\mathrm{d} E) \le \mathcal B(\widetilde E_1)\, \mathcal N_{\infty}^{\mathrm{I}} (\widetilde E_2) \le  c\,\mathcal B(\widetilde E_1)\, \widetilde E_2^{1/2} < \infty \ ,
 \end{align*}
 see also \eqref{Ungleichung Appendix 2 kritische Dichte Energieschranke 2}.
 
In the following, we write $\mathcal N_{\infty}^{\mathrm{I}}(E+)$ for the right-sided limit of $\mathcal N_{\infty}^{\mathrm{I}}$ at $E$; recall that $\mathcal N_{\infty}^{\mathrm{I}}$ is left-continuous.

For the term in line \eqref{Appendix proof 2 critical density align 1 1} we get, using integration by parts, see, e.g., \cite[Theorem 21.67]{hewitt1965real}, and the fact that $\mathcal N_{\infty}^{\mathrm{I}}$ is non-decreasing,
 \begin{align*}
 A &:= \lim\limits_{\epsilon_1 \searrow 0} \int\limits_{(\epsilon_1,\widetilde E_1]} \mathcal{B}(E) \, \mathcal N_{\infty} (\mathrm{d} E)  \\
 & \,\le \lim\limits_{\epsilon_1 \searrow 0} \left\{ \mathcal B(\widetilde E_1) \, \mathcal N_{\infty}^{\mathrm{I}} (\widetilde E_1+) + \int\limits_{\epsilon_1}^{\widetilde E_1} \mathcal N_{\infty}^{\mathrm{I}} (E) \big[\mathcal B(E)\big]^2\, \beta \mathrm{e}^{\beta E} \, \mathrm{d} E \right\}\\
 & \,\le \mathcal B(\widetilde E_1)\, \mathcal N_{\infty}^{\mathrm{I}} (2 \widetilde E_1) + \int\limits_{0}^{\widetilde E_1} \mathcal N_{\infty}^{\mathrm{I}} (E) \big[\mathcal B(E)\big]^2\, \beta \mathrm{e}^{\beta E} \, \mathrm{d} E \\
    \shortintertext{and by Theorem~\ref{Lifshitz Auslaufer one dimensional},}
 & \le \mathcal B(\widetilde E_1)\, \mathcal N_{\infty}^{\mathrm{I}} (2 \widetilde E_1) + \int\limits_{0}^{\widetilde E_1} \widetilde M \mathrm{e}^{- \nu \pi E^{-1/2}} (\beta E)^{-2} \beta \mathrm{e}^{\beta E} \, \mathrm{d} E \\
 & \le \mathcal B(\widetilde E_1)\, \mathcal N_{\infty}^{\mathrm{I}} (2 \widetilde E_1) + \widetilde M \beta^{-1}  \mathrm{e}^{\beta \widetilde E_1} \int\limits_{0}^{\widetilde E_1} \dfrac{4! E^2}{(\nu \pi)^4} E^{-2} \, \mathrm{d} E < \infty \ .
 \end{align*}
  As a last step, we show that also the term in line \eqref{Appendix proof 2 critical density align 1 3} is finite. We obtain
 \begin{align*}
 C & := \lim\limits_{\epsilon_2 \to \infty}  \int\limits_{(\widetilde E_2,\epsilon_2]} \mathcal{B}(E) \, \mathcal N_{\infty} (\mathrm{d} E) \\
   & \, \le \lim\limits_{\epsilon_2 \to \infty} \left[ \mathcal B(\epsilon_2)\, \mathcal N_{\infty}^{\mathrm{I}} (\epsilon_2+) + \int\limits_{\widetilde E_2}^{\epsilon_2} \mathcal N_{\infty}^{\mathrm{I}} (E) \big[\mathcal B(E)\big]^2\, \beta \mathrm{e}^{\beta E} \, \mathrm{d} E \right] 
\end{align*}
by using integration by parts. Since $\mathcal N_{\infty}^{\mathrm{I}}$ is non-decreasing,
\begin{align*} C& \, \le \lim\limits_{\epsilon_2 \to \infty} \left[ \mathcal B(\epsilon_2)\, \mathcal N_{\infty}^{\mathrm{I}} (2\epsilon_2) + \int\limits_{\widetilde E_2}^{\epsilon_2} \mathcal N_{\infty}^{\mathrm{I}} (E) \big[\mathcal B(E)\big]^2\, \beta \mathrm{e}^{\beta E} \, \mathrm{d} E \right] \\
  \shortintertext{and using \eqref{Ungleichung Appendix 2 kritische Dichte Energieschranke 1} and \eqref{Ungleichung Appendix 2 kritische Dichte Energieschranke 2},} 
  & \, \le \lim\limits_{\epsilon_2 \to \infty} \left[ 2 c \, \mathrm{e}^{-\beta  \epsilon_2} \epsilon_2^{1/2} + 2 c \beta \, \int\limits_{\widetilde E_2}^{\epsilon_2} E^{1/2} \mathrm{e}^{-2\beta E } \mathrm{e}^{\beta E} \, \mathrm{d} E \right] < \infty \ .
  \end{align*}
Altogether we have by now proved~\eqref{EquationProofRelationXXXXX}. 

The final statement about $\rho_c(\beta)$ then follows, using again dominated convergence, by
 \begin{align*}
  \rho_c(\beta) & = \sup\limits_{\mu \in (-\infty,0)} \left\{ \int\limits_{\mathds R} \mathcal{B}(E - \mu) \, \mathcal N_{\infty}(\mathrm{d} E) \right\} \\
  & = \sup\limits_{\mu \in (-\infty,0)} \left\{ \int\limits_{(0,\infty)} \mathcal{B}(E - \mu) \, \mathcal N_{\infty}(\mathrm{d} E) \right\} \\
  & = \lim\limits_{\mu \to 0^{-}} \int\limits_{(0,\infty)} \mathcal{B}(E - \mu) \, \mathcal N_{\infty}(\mathrm{d} E) \\
  & = \int\limits_{(0,\infty)} \mathcal{B}(E - \mu) \, \mathcal N_{\infty}(\mathrm{d} E) < \infty\ .
 \end{align*}
 \end{proof}
The statement in the next lemma is not trivial since we have only vague convergence of the density of states and not weak convergence. 
 \begin{lemma} \label{Lemma wegen vager statt schwacher Konvergenz}
  For all $\mu < 0$ we have $\mathds{P}$-almost surely
  \begin{align} \label{equation Lemma wegen vager statt schwacher Konvergenz}
\lim\limits_{N \to \infty} \int\limits_{(0,\infty)} \mathcal{B}(E - \mu) \, \mathcal N_N^{\omega}(\mathrm{d} E) = \int\limits_{(0,\infty)} \mathcal{B}(E - \mu) \, \mathcal N_{\infty}(\mathrm{d} E) \ .    
  \end{align}
 \end{lemma}
\begin{proof}
 Let $\mu < 0$ be given. Then, for all $E_2 > 0$ we get
 \begin{align*}
 \begin{split}
 \limsup\limits_{N \to \infty} \int\limits_{(0,\infty)} &\mathcal{B}(E - \mu) \, \mathcal N_N^{\omega}(\mathrm{d} E) \\
 \le \, & \limsup\limits_{N \to \infty} \int\limits_{(0,E_2]} \mathcal{B}(E - \mu) \, \mathcal N_N^{\omega}(\mathrm{d} E) + \limsup\limits_{N \to \infty} \int\limits_{[E_2,\infty)} \mathcal{B}(E - \mu) \, \mathcal N_N^{\omega}(\mathrm{d} E) \ .
 \end{split}
 \end{align*}
 For $0 < \mathcal E < - \mu$ we define the real function $g_{\mathcal E, E_2}$ as
 \begin{align} \label{definition g mathcal E E_2 E}
 g_{\mathcal E, E_2} (E) := \begin{cases}
                        0 & \quad \text{ if } E \le - \mathcal E \\
                        1 + \mathcal E^{-1} E & \quad \text{ if } - \mathcal E < E < 0 \\
                        1 & \quad \text{ if } 0 \le E \le E_2 \\
                        1 - \dfrac{E - E_2}{\mathcal E} & \quad \text{ if } E_2 < E < E_2 + \mathcal E \\
                        0 & \quad \text{ if } E > E_2 + \mathcal E
                       \end{cases}\ .
\end{align}
For the first integral we $\mathds{P}$-almost surely obtain, by $\mathds{P}$-almost sure vague convergence,
 \begin{align*}
   \limsup\limits_{N \to \infty} \int\limits_{(0,E_2]} &\mathcal{B}(E - \mu) \, \mathcal N_N^{\omega}(\mathrm{d} E) \\
  \le \, & \limsup\limits_{N \to \infty} \int\limits_{\mathds R} g_{\mathcal E, E_2} (E) \mathcal{B}(E - \mu) \, \mathcal N_N^{\omega}(\mathrm{d} E) \\
  \le \, & \mathcal{B}(E_2 - \mu) \mathcal N_{\infty}^{\mathrm{I}}(E_2 + \mathcal E)+ \int\limits_{(0,E_2]}  \mathcal{B}(E - \mu) \mathcal N_{\infty}(\mathrm{d} E) \ .
 \end{align*}
Furthermore, 
 \begin{align*}
   \limsup\limits_{N \to \infty} &\int\limits_{[E_2,\infty)}\mathcal{B}(E - \mu) \, \mathcal N_N^{\omega}(\mathrm{d} E) \le \limsup\limits_{N \to \infty} \lim\limits_{\epsilon_2 \to \infty} \int\limits_{[E_2,\epsilon_2]} \mathcal{B}(E - \mu) \, \mathcal N_N^{\omega}(\mathrm{d} E)\ , \\
  \shortintertext{and integrating by parts gives}
  \le \, & \limsup\limits_{N \to \infty} \lim\limits_{\epsilon_2 \to \infty} \left[ \mathcal{B}(\epsilon_2- \mu) \mathcal N_N^{\mathrm{I},\omega}(\epsilon_2) +  \beta\int\limits_{E_2}^{\epsilon_2} \mathcal N_N^{\mathrm{I},\omega}(E) \big[\mathcal B(E - \mu)\big]^2\, \mathrm{e}^{\beta (E - \mu)}  \, \mathrm{d} E\right] \ .
  \end{align*}
If we denote with $\mathcal N_N^{\mathrm{I},(0)}$ the integrated density of states of the free Hamiltonian $-\ud^2/\ud x^2$ on $H^1_0(\Lambda_N)$ then we can further bound this by
  \begin{align*}
  \le \, & \limsup\limits_{N \to \infty} \lim\limits_{\epsilon_2 \to \infty} \left[ \mathcal{B}(\epsilon_2 - \mu) \mathcal N_N^{\mathrm{I},(0)}(\epsilon_2) + \beta\int\limits_{E_2}^{\epsilon_2} \mathcal N_N^{\mathrm{I},(0)}(E) \big[\mathcal B(E - \mu)\big]^2\, \mathrm{e}^{\beta (E - \mu)}   \, \mathrm{d} E\right]\ .
  \end{align*}
Since $\mathcal N_N^{\mathrm{I},(0)}(E) \le \pi^{-1} E^{1/2}$ for all $E \ge 0$ and all $N \in \mathds N$ we get
  \begin{align*}
  = \, & \lim\limits_{\epsilon_2 \to \infty} \left[ \pi^{-1} \mathcal{B}(\epsilon_2 - \mu) \epsilon_2^{1/2} + \beta\pi^{-1}\int\limits_{E_2}^{\epsilon_2}  E^{1/2} \big[\mathcal B(E - \mu)\big]^2\, \mathrm{e}^{\beta (E - \mu)}   \, \mathrm{d} E \right]\\
  = \, & \beta \pi^{-1}\int\limits_{E_2}^{\infty}  E^{1/2} \big[\mathcal B(E - \mu)\big]^2\, \mathrm{e}^{\beta (E - \mu)}  \, \mathrm{d} E \ .
 \end{align*}
Hence, in total we obtain $\mathds{P}$-almost surely
\begin{align*}
 \limsup\limits_{N \to \infty} &\int\limits_{(0,\infty)} \mathcal{B}(E - \mu) \, \mathcal N_N^{\omega}(\mathrm{d} E) \\
\le \, & \lim\limits_{E_2 \to \infty} \mathcal{B}(E_2 - \mu) \mathcal N_{\infty}^{\mathrm{I}}(E_2 + \mathcal E)+\lim\limits_{E_2 \to \infty} \int\limits_{(0,E_2]} \mathcal{B}(E - \mu) \mathcal N_{\infty}(\mathrm{d} E)  \\
& + \, \beta \pi^{-1} \lim\limits_{E_2 \to \infty}  \int\limits_{E_2}^{\infty}  E^{1/2} \big[\mathcal B(E - \mu)\big]^2\, \mathrm{e}^{\beta (E - \mu)}  \, \mathrm{d} E \\
= \, & \int\limits_{(0,\infty)}  \mathcal{B}(E - \mu) \mathcal N_{\infty}(\mathrm{d} E) \ ,
\end{align*}
where we also used Theorem~\ref{Lifshitz Auslaufer one dimensional}.

On the other hand, for all $E_2 > 0$ $\mathds{P}$-almost surely, 
 \begin{align*}
 \liminf\limits_{N \to \infty} \int\limits_{(0,\infty)} \mathcal{B}(E - \mu) \, \mathcal N_N^{\omega}(\mathrm{d} E) & \ge \liminf\limits_{N \to \infty} \int\limits_{\mathds R} g_{\mathcal E, E_2}(E) \mathcal{B}(E - \mu) \, \mathcal N_N^{\omega}(\mathrm{d} E) \\
 & = \int\limits_{\mathds R} g_{\mathcal E, E_2}(E) \mathcal{B}(E - \mu) \, \mathcal N_{\infty}(\mathrm{d} E) \\
 & \ge \int\limits_{(0,E_2]} \mathcal{B}(E - \mu) \, \mathcal N_{\infty}(\mathrm{d} E) \ .
 \end{align*}
Thus, $\mathds{P}$-almost surely,
\begin{align*}
 \liminf\limits_{N \to \infty} \int\limits_{(0,\infty)} \mathcal{B}(E - \mu) \, \mathcal N_N^{\omega}(\mathrm{d} E) & \ge \lim\limits_{E_2 \to \infty} \int\limits_{(0,E_2]} \mathcal{B}(E - \mu) \, \mathcal N_{\infty}(\mathrm{d} E) \,
 =\, \int\limits_{(0,\infty)} \mathcal{B}(E - \mu)\, \mathcal N_{\infty}(\mathrm{d} E) \ .
\end{align*}
\end{proof}
The next lemma is taken from \cite{lenoble2004bose}. We follow in parts their proof. 
 \begin{lemma} \label{Satz Konvergenzverhalten mu}
If $\rho < \rho_c(\beta)$, then $\mu^{\omega}_{N}$ converges $\mathds{P}$-almost surely to a non-random limit point $\widehat \mu < 0$. On the other hand, if $\rho \geq \rho_c(\beta)$, then $\mu^{\omega}_{N}$ converges $\mathds{P}$-almost surely to $0$.
\end{lemma}
\begin{proof} In a first step, we show that the sequence $(\mu^{\omega}_{N})_{N=1}^{\infty}$ has $\mathds{P}$-almost surely at least one accumulation point in both cases: Note that 
\begin{align} \label{assumption limsup e beta E finite}
 \limsup\limits_{N \to \infty} \int\limits_{(0,\infty)} \mathrm{e}^{-\beta E} \mathcal N_N^{\omega}(\mathrm{d} E)  \le \limsup\limits_{N \to \infty}\frac{1}{L_N} \sum\limits_{j =1}^{\infty} \mathrm{e}^{-\beta (j\pi / L_N)^2} = (4\pi \beta)^{-1/2} < \infty
\end{align}
for $\mathds{P}$-almost all $\omega \in \Omega$. We define
\begin{align*}
\phi_N^{\omega}(\beta) := \dfrac{1}{L_N} \sum\limits_{j=1}^{\infty} \mathrm{e}^{-\beta E_N^{j,\omega}} = \int\limits_{(0,\infty)} \mathrm{e}^{-\beta E} \, \mathcal N_N^{\omega}(\mathrm{d} E)\ .
\end{align*}
The relation between $\phi_N^{\omega}(\beta)$ and $\rho$ is simply
\begin{align*}
\rho & = \dfrac{1}{L_N} \sum\limits_{j = 1}^{\infty} \left( \mathrm{e}^{\beta ( E_N^{j,{\omega}} - \mu_N^{\omega})} -1 \right)^{-1} = 
 \dfrac{1}{L_N} \sum\limits_{j= 1}^{\infty} \mathrm{e}^{-\beta E_N^{j,\omega}}
\dfrac{1}
{
\mathrm{e}^{- \beta \mu_N^{\omega}} - \mathrm{e}^{-\beta E_N^{j,\omega}} }  \\
& \le \left( \dfrac{1}{L_N} \sum\limits_{j = 1}^{\infty} \mathrm{e}^{- \beta E_N^{j,{\omega}}} \right) \dfrac{\mathrm{e}^{\beta \mu_N^{\omega}}}{1 - \mathrm{e}^{-\beta (E_N^{1,{\omega}} - \mu_N^{\omega})}} = \phi_N^{\omega}(\beta) \dfrac{\mathrm{e}^{\beta \mu_N^{\omega}}}{1 - \mathrm{e}^{-\beta (E_N^{1,{\omega}} - \mu_N^{\omega})}} \ .
\end{align*}
Consequently, we conclude that
\begin{align*}
\rho - \rho \mathrm{e}^{-\beta E_N^{1,\omega}} \mathrm{e}^{\beta \mu_N^{\omega}} \le \phi_N^{\omega}(\beta) \mathrm{e}^{\beta \mu_N^{\omega}}
\end{align*}
and
\begin{align}        
\beta^{-1} \ln \left( \dfrac{\rho}{\phi_N^{\omega}(\beta) + \rho \mathrm{e}^{-\beta E_N^{1,\omega}}} \right) & \le \mu_N^{\omega} \ .
\end{align}

Due to~\eqref{assumption limsup e beta E finite} we have $\limsup_{N \to \infty} \phi^{\omega}_N(\beta) < \infty$. Moreover, since $\mu_N^{\omega}< E_N^{1,\omega}$ for all $N \in \mathds N$ and since $E_N^{1,\omega}$ converges $\mathds{P}$-almost surely to $0$, we obtain $\mathds{P}$-almost surely
\begin{align*}
\beta^{-1} \ln \left( \dfrac{\rho}{\limsup\limits_{N \to \infty} \phi_N^{\omega}(\beta) + \rho } \right) \le \liminf\limits_{N \to \infty} \mu_N^{\omega}\le \limsup\limits_{N \to \infty} \mu_N^{\omega}\le 0 \ .
\end{align*}
This shows that there exists a set $\widetilde \Omega \subset \Omega$ with measure $\mathds P(\widetilde \Omega) = 1$ such that for all $\omega \in \widetilde \Omega$ equation \eqref{equation Lemma wegen vager statt schwacher Konvergenz} from Lemma~\ref{Lemma wegen vager statt schwacher Konvergenz} holds and the sequence $(\mu^{\omega}_{N})_{N=1}^{\infty}$ has at least one accumulation point $\mu_{\infty}^{\omega} \in \mathds R$.

Now, for $\omega \in \widetilde \Omega$ consider the case where $\rho < \rho_c(\beta)$: Since
\begin{align} \label{function mu to int e E mu}
(-\infty,0) \to \mathds{R}, \quad \mu \mapsto \int\limits_{(0,\infty)} \mathcal{B}(E - \mu) \, \mathcal N_{\infty} (\mathrm{d} E) 
\end{align}
is a strictly increasing function, there is a unique, non-random solution $\widehat \mu < 0$ to 
\begin{align} \label{definierende Gleichung fuer mu hat}
 \rho = \int\limits_{(0,\infty)} \mathcal{B}(E - \mu) \, \mathcal N_{\infty} (\mathrm{d} E) \ .
\end{align}
Suppose to the contrary that the accumulation point $\mu_{\infty}^{\omega}$ is zero. Then there exists a subsequence $(\mu_{N_j}^{\omega})_{j=1}^{\infty}$ which converges to $0$.
It follows that 
\begin{align}
\dfrac{\widehat \mu }{2} < \mu_{N_j}^{\omega}
\end{align}
  for all but finitely many $j \in \mathds{N}$. Because 
  \begin{align}
(-\infty,0) \to \mathds{R}, \quad \mu \mapsto \int\limits_{(0,\infty)} \mathcal{B}(E - \mu) \, \mathcal N_N^{\omega}(\mathrm{d} E)
  \end{align}
  is strictly increasing, 
  \begin{align} \label{Ergebnis Widerspruch Konvergenz von mu}
 \int\limits_{(0,\infty)} \mathcal B(E-\widehat{\mu}/2)\, \mathcal N_{N_j}^{\omega}(\mathrm{d} E) < \int\limits_{(0,\infty)} \mathcal B(E-\mu_{N_j}^{\omega}) \, \mathcal N_{N_j}^{\omega}(\mathrm{d} E) = \rho
\end{align} 
  for all but finitely many $j \in \mathds N$. However, due to~\eqref{definierende Gleichung fuer mu hat} and employing Lemma~\ref{Lemma wegen vager statt schwacher Konvergenz}
 %
 \begin{align}\label{HP mu Loesung der Gleichung}
 \begin{split}
 \rho & = \int\limits_{(0,\infty)} \mathcal{B}(E - \widehat \mu) \, \mathcal N_{\infty}(\mathrm{d} E) \\
 & < \int\limits_{(0,\infty)} \mathcal B(E-\widehat{\mu}/2) \, \mathcal N_{\infty} (\mathrm{d} E) \\
 & =  \lim\limits_{N \to \infty} \int\limits_{(0,\infty)} \mathcal B(E-\widehat{\mu}/2)\, \mathcal N_{N}^{\omega} \, (\mathrm{d} E) \ ,
 \end{split}
 \end{align}
  which is a contradiction to~\eqref{Ergebnis Widerspruch Konvergenz von mu}. As a consequence, any accumulation point $\mu_{\infty}^{\omega}$ is strictly smaller than $0$. In addition, using Lemma~\ref{Lemma wegen vager statt schwacher Konvergenz}, $m \in \mathds{N}$, 
  \begin{align*}
  \limsup\limits_{j \to \infty} \int\limits_{(0,\infty)} \mathcal B(E-\mu_{N_j}^\omega) \mathcal N_{N_j}^{\omega}(\mathrm{d}E) & \le \lim\limits_{j \to \infty} \int\limits_{(0,\infty)} \mathcal B(E - \mu_{\infty}^{\omega} + m^{-1} \mu_{\infty}^{\omega}/2) \, \mathcal N_{N_j}^{\omega}(\mathrm{d}E) \\
  & = \int\limits_{(0,\infty)} \mathcal B(E - \mu_{\infty}^{\omega} + m^{-1}\mu_{\infty}^{\omega}/2) \, \mathcal N_{\infty}(\mathrm{d}E)
  \end{align*}
and
\begin{align*}
  \liminf\limits_{j \to \infty} \int\limits_{(0,\infty)} \mathcal B(E - \mu_{N_j}^{\omega}) \, \mathcal N_{N_j}^{\omega}(\mathrm{d}E) & \ge \lim\limits_{j \to \infty} \int\limits_{(0,\infty)} \mathcal B(E - \mu_{\infty}^{\omega} - m^{-1} \mu_{\infty}^{\omega}/2) \, \mathcal N_{N_j}^{\omega}(\mathrm{d}E) \\
  & = \int\limits_{(0,\infty)} \mathcal B(E - \mu_{\infty}^{\omega} - m^{-1} \mu_{\infty}^{\omega}/2)\, \mathcal N_{\infty}(\mathrm{d}E)\ .
  \end{align*}
Hence, since $m \in \mathds{N}$ was arbitrary,
 %
 \begin{align*}
  \lim\limits_{j \to \infty} \int\limits_{(0,\infty)} \mathcal B(E - \mu_{N_j}^{\omega}) \, \mathcal N_{N_j}^{\omega}(\mathrm{d}E)
  & = \int\limits_{(0,\infty)} \mathcal B(E - \mu_{\infty}^{\omega}) \, \mathcal N_{\infty}(\mathrm{d}E) \ .
 \end{align*}
 We conclude that
  \begin{align} \label{Gleichung mu Haufungspunkt gleich mu tilde}
  \begin{split}
  \int\limits_{(0,\infty)} \mathcal B(E - \widehat \mu) \, \mathcal N_{\infty}(\mathrm{d} E) & = \rho = \lim\limits_{j \to \infty} \int\limits_{(0,\infty)} \mathcal B(E - \mu_{N_j}^{\omega}) \, \mathcal N_{N_j}^{\omega}(\mathrm{d} E) \\
  &= \int\limits_{(0,\infty)} \mathcal B(E - \mu_{\infty}^{\omega}) \, \mathcal N_{\infty}(\mathrm{d} E)\ ,
  \end{split}
 \end{align}
 holds $\mathds{P}$-almost surely for all convergent subsequences of $(\mu^{\omega}_{N})_{N=1}^{\infty}$ with corresponding limit point $\mu_{\infty}^{\omega}$. Hence, due to the strict monotonicity of the function~\eqref{function mu to int e E mu} and due to~\eqref{Gleichung mu Haufungspunkt gleich mu tilde}, any accumulation point $\mu_{\infty}^{\omega}$ is equal to $\widehat \mu$. In other words, the sequence $(\mu^{\omega}_{N})_{N=1}^{\infty}$ converges to the non-random limit $\widehat \mu < 0$.
 
In the next step we assume that $\rho \ge \rho_c(\beta)$: Suppose to the contrary that the sequence $(\mu^{\omega}_{N})_{N=1}^{\infty}$ has an accumulation point $\mu_{\infty}^{\omega} < 0$ with the subsequence $(\mu_{N_j}^{\omega})_{j=1}^{\infty}$ converging to it. As in~\eqref{Gleichung mu Haufungspunkt gleich mu tilde} we get
 \begin{align*}
 \rho = \lim\limits_{j \to \infty} \int\limits_{(0,\infty)} \mathcal B(E - \mu_{N_j}^\omega) \, \mathcal N_{N_j}^{\omega}(\mathrm{d} E) = \int\limits_{(0,\infty)} \mathcal B(E - \mu_{\infty}^{\omega}) \, \mathcal N_{\infty} (\mathrm{d} E) \ge \rho_c(\beta) \ .
 \end{align*}
 However, one also has
 \begin{align*}
 \int\limits_{(0,\infty)} \mathcal B(E - \mu_{\infty}^{\omega}) \, \mathcal N_{\infty}(\mathrm{d} E) < \sup\limits_{\mu \in (-\infty, 0)} \left\{ \int\limits_{(0,\infty)} \mathcal B(E - \mu) \, \mathcal N_{\infty}(\mathrm{d} E) \right\} =  \rho_c(\beta) \ ,
 \end{align*}
 yielding a contradiction. Since this holds for any subsequence, we conclude the statement. 
\end{proof}
The next lemma is essential in the proof of generalized BEC in the supercritical region $\rho > \rho_c(\beta)$. We do not know whether the limit of $\int\limits_{(\epsilon,\infty)} \mathcal B(E-\mu_N^{\omega}) \, \mathcal N_N^{\omega} (\mathrm{d} E)$ as $N\to\infty$ exists $\mathds{P}$-almost surely. Therefore we state bounds on the $\limsup$ and $\liminf$, which, most importantly, coincide in the limit $\epsilon\searrow 0$.

\begin{lemma} \label{Lemma beweis lim int epsilon infty und lim lim int epsilion infty rhoc vage Konvergenz} 
If $\rho \ge \rho_c(\beta)$ and $\epsilon > 0$, then $\mathds{P}$-almost surely, 
 \begin{align}
  \limsup\limits_{N \to \infty} \int\limits_{(\epsilon,\infty)} \mathcal B(E-\mu_N^{\omega}) \, \mathcal N_N^{\omega} (\mathrm{d} E) & \le \int\limits_{(\epsilon,\infty)} \mathcal B(E) \, \mathcal N_{\infty} (\mathrm{d} E) + \dfrac{2}{\beta \epsilon} \mathcal N_{\infty}^{\mathrm{I}}(\epsilon) \ , \label{LemmaA6limsup} \\
   \liminf\limits_{N \to \infty} \int\limits_{(\epsilon,\infty)} \mathcal B(E-\mu_N^{\omega}) \, \mathcal N_N^{\omega} (\mathrm{d} E) & \ge \int\limits_{(\epsilon,\infty)} \mathcal B(E) \, \mathcal N_{\infty} (\mathrm{d} E) - \dfrac{4}{\beta \epsilon} \mathcal N_{\infty}^{\mathrm{I}}(2\epsilon) \nonumber\ ,
 \end{align}
and
 \begin{align} \label{LemmaA6limsup rho_c}
 \begin{split}
   &\lim\limits_{\epsilon \searrow 0} \limsup\limits_{N \to \infty} \int\limits_{(\epsilon,\infty)} \mathcal B(E-\mu_N^{\omega}) \, \mathcal N_N^{\omega} (\mathrm{d} E) \\
   &\quad = \, \lim\limits_{\epsilon \searrow 0} \liminf\limits_{N \to \infty} \int\limits_{(\epsilon,\infty)} \mathcal B(E-\mu_N^{\omega}) \, \mathcal N_N^{\omega} (\mathrm{d} E)  = \rho_c(\beta) \ .
   \end{split}
 \end{align}
\end{lemma}
\begin{proof}
%
In a first step we note that $\mathds{P}$-almost surely and for all $E_2 > \epsilon > 0$
\begin{align} \label{lim int NN int N infty}
\lim\limits_{N \to \infty} \int\limits_{\mathds R} g_{\epsilon}^{(E_2)}(E) \mathcal B(E-\mu_N^{\omega}) \, \mathcal N_N^{\omega} (\mathrm{d} E) = \int\limits_{\mathds R} g_{\epsilon}^{(E_2)}(E) \mathcal B(E) \, \mathcal N_{\infty} (\mathrm{d} E) 
\end{align}
where 
 \begin{align} 
 g_{\epsilon}^{(E_2)}(E):= \begin{cases}
 0 \quad & \text{ if } E < \epsilon/2 \\
 \dfrac{E - \epsilon/2}{\epsilon/2} & \text{ if } \epsilon/2 \le E \le \epsilon \\
 1 \quad & \text{ if } \epsilon < E \le E_2 \\
 1 - \dfrac{E - E_2}{E_2} \quad & \text{ if } E_2  < E < 2 E_2 \\
 0 \quad & \text{ if } E \ge 2 E_2
 \end{cases} \ .
 \end{align} 

This can be shown as follows: One has
 \begin{align} \label{Gleichung beweis verallg BEC 2}
  & \left|\, \int\limits_{\mathds R} g_{\epsilon}^{(E_2)}(E) \mathcal B(E-\mu_N^{\omega}) \, \mathcal N_N^{\omega} (\mathrm{d} E) - \int\limits_{\mathds R} g_{\epsilon}^{(E_2)}(E) \mathcal B(E) \, \mathcal N_{\infty} (\mathrm{d} E) \, \right| \\
 &\quad \le  \, \left| \, \int\limits_{\mathds R}  g_{\epsilon}^{(E_2)}(E) \left( \mathcal B(E-\mu_N^{\omega}) - \mathcal B(E) \right) \, \mathcal N_N^{\omega}(\mathrm{d} E) \, \right| \label{Gleichung beweis verallg BEC 2 line 1} \\
 & \qquad + \, \left|\,  \int\limits_{\mathds R}  g_{\epsilon}^{(E_2)}(E) \mathcal B(E) \mathbb\, \big[ \mathcal N_N^{\omega}(\mathrm{d} E) - \mathcal N_{\infty}(\mathrm{d} E) \big] \, \right|  \label{Gleichung beweis verallg BEC 2 line 2} \ .
  \end{align}
For the term in line~\eqref{Gleichung beweis verallg BEC 2 line 1} we get $\mathds{P}$-almost surely, using that $\mathds{P}$-almost surely $\mu_N^{\omega}$ converges to $0$ and \eqref{Gleichung Bedingung mu aequivalent},
 \begin{align*}
 & \lim\limits_{N \to \infty} \left| \, \int\limits_{\mathds R} g_{\epsilon}^{(E_2)}(E)\dfrac{\mathrm{e}^{\beta E } - \mathrm{e}^{\beta ( E - \mu_N^{\omega} ) }}{\left( \mathrm{e}^{\beta ( E - \mu_N^{\omega} ) } - 1 \right) \cdot \Big( \mathrm{e}^{\beta E } - 1 \Big)} \, \mathcal N_N^{\omega}(\mathrm{d} E) \, \right|
 = 0 \ .
 \end{align*}
The term in line \eqref{Gleichung beweis verallg BEC 2 line 2} $\mathds{P}$-almost surely converges to zero for $N \to \infty$ by vague convergence.

Next, we obtain, for all $E_2 > \epsilon$ and all $N \in \mathds N$,
 \begin{align*}
   \int\limits_{(\epsilon,\infty)} &\mathcal B(E-\mu_N^{\omega}) \, \mathcal N_N^{\omega} (\mathrm{d} E) \\
  = \, & \int\limits_{(\epsilon,E_2]} \mathcal B(E-\mu_N^{\omega})\, \mathcal N_N^{\omega} (\mathrm{d} E) + \int\limits_{(E_2,\infty)} \mathcal B(E-\mu_N^{\omega}) \, \mathcal N_N^{\omega} (\mathrm{d} E) \ .
  \end{align*}
 On the one hand, by \eqref{lim int NN int N infty} we have $\mathds{P}$-almost surely,  
  \begin{align*}
  \limsup\limits_{N \to \infty} \int\limits_{(\epsilon,E_2]} \mathcal B(E-\mu_N^{\omega}) \, \mathcal N_N^{\omega} (\mathrm{d} E) & \le \lim\limits_{N \to \infty} \int\limits_{\mathds R} g_{\epsilon}^{(E_2)}(E) \mathcal B(E-\mu_N^{\omega}) \, \mathcal N_N^{\omega} (\mathrm{d} E) \\
  & = \int\limits_{\mathds R} g_{\epsilon}^{(E_2)}(E) \mathcal B(E) \, \mathcal N_{\infty} (\mathrm{d} E) \ .
  \end{align*}
 As for the second integral, we obtain $\mathds{P}$-almost surely,
 \begin{align*}
 \begin{split}
  \limsup\limits_{N \to \infty} \int\limits_{(E_2,\infty)} \mathcal B(E-\mu_N^{\omega}) \, \mathcal N_N^{\omega} (\mathrm{d} E) & \le \limsup\limits_{N \to \infty} \int\limits_{[E_2,\infty)} \mathcal B(E-\epsilon/2) \, \mathcal N_N^{\omega} (\mathrm{d} E) \\
  & \le \beta\pi^{-1} \int\limits_{E_2}^{\infty} E^{1/2} \big[\mathcal B(E-\epsilon/2)\big]^{2} \mathrm{e}^{\beta(E - \epsilon/2)}  \, \mathrm{d} E
  \end{split}
 \end{align*}
 where the last step is as in the proof of Lemma~\ref{Lemma wegen vager statt schwacher Konvergenz} (with $\mu = \epsilon/2$).

We conclude that 
\begin{align*}
 \limsup\limits_{N \to \infty} &\int\limits_{(\epsilon,\infty)}\mathcal B(E-\mu_N^{\omega}) \, \mathcal N_N^{\omega} (\mathrm{d} E) \\
 &\le \int\limits_{\mathds R} g_{\epsilon}^{(E_2)}(E) \mathcal B(E) \, \mathcal N_{\infty} (\mathrm{d} E) + \beta\pi^{-1/2}\int\limits_{E_2}^{\infty}  E^{1/2} \big[\mathcal B(E-\epsilon/2)\big]^2 \mathrm{e}^{\beta(E - \epsilon/2)}  \, \mathrm{d} E
\end{align*}
for all $E_2 > 0$ and hence
\begin{align*}
 \limsup\limits_{N \to \infty} &\int\limits_{(\epsilon,\infty)} \mathcal B(E-\mu_N^{\omega}) \, \mathcal N_N^{\omega} (\mathrm{d} E) \\
 \le \, & \lim\limits_{E_2 \to \infty} \int\limits_{\mathds R} g_{\epsilon}^{(E_2)}(E) \mathcal B(E) \, \mathcal N_{\infty} (\mathrm{d} E) \\
 & + \, \lim\limits_{E_2 \to \infty} \beta\pi^{-1/2}\int\limits_{E_2}^{\infty}  E^{1/2} \big[\mathcal B(E-\epsilon/2)\big]^2 \mathrm{e}^{\beta(E - \epsilon/2)}  \, \mathrm{d} E \\
 \le \, & \lim\limits_{E_2 \to \infty} \int\limits_{[\epsilon/2,2E_2]} \mathcal B(E) \, \mathcal N_{\infty} (\mathrm{d} E) \\ 
 \le \, & \int\limits_{(\epsilon,\infty)} \mathcal B(E) \, \mathcal N_{\infty} (\mathrm{d} E) + \dfrac{2}{\beta \epsilon} \mathcal N_{\infty}^{\mathrm{I}}(\epsilon) \ .
\end{align*}

On the other hand, for all $E_2 > \epsilon$ and all $N \in \mathds N$,
\begin{align*}
   \int\limits_{(\epsilon,\infty)} &\mathcal B(E-\mu_N^{\omega}) \, \mathcal N_N^{\omega} (\mathrm{d} E) \\
   \ge \, & \int\limits_{\mathds R} g_{\epsilon}^{(E_2)}(E) \mathcal B(E-\mu_N^{\omega}) \, \mathcal N_N^{\omega} (\mathrm{d} E) - \int\limits_{[\epsilon/2,\epsilon]} g_{\epsilon}^{(E_2)}(E) \mathcal B(E-\mu_N^{\omega}) \, \mathcal N_N^{\omega} (\mathrm{d} E) \ .
  \end{align*}
  For the first integral, by \eqref{lim int NN int N infty}, $\mathds{P}$-almost surely
  \begin{align*}
   \lim\limits_{N \to \infty} \int\limits_{\mathds R} g_{\epsilon}^{(E_2)}(E) \mathcal B(E-\mu_N^{\omega}) \, \mathcal N_N^{\omega} (\mathrm{d} E) & = \int\limits_{\mathds R} g_{\epsilon}^{(E_2)}(E) \mathcal B(E) \, \mathcal N_{\infty} (\mathrm{d} E) \\
   & \ge \int\limits_{(\epsilon,E_2]} \mathcal B(E) \, \mathcal N_{\infty} (\mathrm{d} E) \ .
  \end{align*}
Since $\mu_N^{\omega}$ converges $\mathds{P}$-almost surely to zero and $\, \mathcal N_N^{\omega}$ converges $\mathds{P}$-almost surely vaguely to $\mathcal N_{\infty}$, the second integral $\mathds{P}$-almost surely converges to
  \begin{align*}
   \limsup\limits_{N \to \infty} &\int\limits_{[\epsilon/2,\epsilon]} g_{\epsilon}^{(E_2)}(E) \mathcal B(E-\mu_N^{\omega}) \, \mathcal N_N^{\omega} (\mathrm{d} E) \\
   \le \, & \limsup\limits_{N \to \infty} \int\limits_{[\epsilon/2,\epsilon]} g_{\epsilon}^{(E_2)}(E) \mathcal B(E-\epsilon/4)  \, \mathcal N_N^{\omega} (\mathrm{d} E) \\
   \le \, & \mathcal B(\epsilon/4) \int\limits_{\mathds R} g_{\epsilon}^{(\epsilon)}(E) \, \mathcal N_{\infty} \, (\mathrm{d} E) \\
   \le \, & \dfrac{4}{\beta \epsilon} \mathcal N_{\infty}^{\mathrm{I}}(2 \epsilon) \ .
  \end{align*}
  We conclude that, $\mathds{P}$-almost surely,
\begin{align*}
   \liminf\limits_{N \to \infty} \int\limits_{(\epsilon,\infty)} \mathcal B(E-\mu_N^{\omega}) \, \mathcal N_N^{\omega} (\mathrm{d} E) \ge \int\limits_{(\epsilon,E_2]} \mathcal B(E) \, \mathcal N_{\infty} (\mathrm{d} E) -  \dfrac{4}{\beta \epsilon} \mathcal N_{\infty}^{\mathrm{I}}(2 \epsilon)
   \end{align*}
for all $E_2 > \epsilon$ and hence
\begin{align*}
   \liminf\limits_{N \to \infty} \int\limits_{(\epsilon,\infty)} \mathcal B(E-\mu_N^{\omega})\, \mathcal N_N^{\omega} (\mathrm{d} E) \ge \int\limits_{(\epsilon,\infty)} \mathcal B(E) \, \mathcal N_{\infty} (\mathrm{d} E) -  \dfrac{4}{\beta \epsilon} \mathcal N_{\infty}^{\mathrm{I}}(2 \epsilon) \ .
   \end{align*}
Finally, with Theorem~\ref{Lifshitz Auslaufer one dimensional}, we obtain
 \begin{align*}
 \lim\limits_{\epsilon \searrow 0} \dfrac{4}{\beta \epsilon} \mathcal N_{\infty}^{\mathrm{I}}(2 \epsilon) \le \lim\limits_{\epsilon \searrow 0} \dfrac{4}{\beta \epsilon} \widetilde M \dfrac{4! (2 \epsilon)^2}{(\pi\nu)^4} = 0 \ ,
 \end{align*}
which, taking Lemma~\ref{finite critical density in LSmodel} into account, shows that $\mathds{P}$-almost surely
 \begin{align*}
 \begin{split}
 \lim\limits_{\epsilon \searrow 0} \limsup\limits_{N \to \infty} \int\limits_{(\epsilon,\infty)} \mathcal B(E-\mu_N^{\omega}) \, \mathcal N_N^{\omega} (\mathrm{d} E) & = \lim\limits_{\epsilon \searrow 0} \liminf\limits_{N \to \infty} \int\limits_{(\epsilon,\infty)} \mathcal B(E-\mu_N^{\omega}) \, \mathcal N_N^{\omega} (\mathrm{d} E) \\
 & = \lim\limits_{\epsilon \searrow 0} \int\limits_{(\epsilon,\infty)} \mathcal B(E) \, \mathcal N_{\infty} (\mathrm{d} E) \\
 & = \rho_c(\beta) \ .
 \end{split}
 \end{align*}
\end{proof}

Finally, we present a proof of generalized BEC and we follow in parts the proof in \cite{lenoble2004bose}.	

\begin{proof}[Proof of Theorem \ref{TheoremGeneralizedBEC}] Assume firstly that $\rho \ge \rho_c(\beta)$: According to Lemma \ref{Satz Konvergenzverhalten mu}, the sequence $(\mu^{\omega}_{N})_{N=1}^{\infty}$ converges to $0$ $\mathds{P}$-almost surely. Also, for all $\epsilon > 0$, all $N \in \mathds N$ and $\mathds{P}$-almost all $\omega \in \Omega$ one has
 \begin{align*}
 \rho = \int\limits_{(0,\epsilon]} \mathcal B(E-\mu_N^{\omega}) \,\mathcal N_N^{\omega} (\mathrm{d} E) + \int\limits_{(\epsilon,\infty)} \mathcal B(E-\mu_N^{\omega}) \,\mathcal N_N^{\omega} (\mathrm{d} E) 
 \end{align*}
 and hence
  \begin{align} \label{Gleichung beweis verallg BEC}
 \begin{split}
 \rho_0 (\beta) & = \lim\limits_{\epsilon \searrow 0} \liminf\limits_{N \to \infty} \int\limits_{(0,\epsilon]} \mathcal B(E-\mu_N^{\omega}) \,\mathcal N_N^{\omega} (\mathrm{d} E) \\
 & = \rho -  \lim\limits_{\epsilon \searrow 0} \limsup\limits_{N \to \infty} \int\limits_{(\epsilon,\infty)} \mathcal B(E-\mu_N^{\omega}) \,\mathcal N_N^{\omega} (\mathrm{d} E) \ .
 \end{split}
 \end{align}
 Consequently, by~\eqref{LemmaA6limsup rho_c}, 
 \begin{align*}
 \rho_0(\beta) = \rho - \rho_c(\beta)
 \end{align*}
  holds $\mathds{P}$-almost surely, implying the statement.
 
 
In a next step assume that $\rho < \rho_c(\beta)$: According to Lemma~\ref{Satz Konvergenzverhalten mu}, the sequence $(\mu^{\omega}_{N})_{N=1}^{\infty}$ $\mathds{P}$-almost surely converges to a limit $\widehat \mu < 0$. Consequently, $\mathds{P}$-almost surely there exists a $\delta > 0$ such that $- \mu_N^{\omega} > \delta$ for all but finitely many $N \in \mathds N$. Hence, $\mathds{P}$-almost surely, with $g_{\epsilon, \epsilon}(E)$ defined as in \eqref{definition g mathcal E E_2 E},
 \begin{align*}
\lim\limits_{\epsilon \searrow 0} \lim\limits_{N \to \infty} \dfrac{1}{N} \sum\limits_{j : E_N^{j,\omega} \le \epsilon} n_N^{j,\omega} & = \rho^{-1} \lim\limits_{\epsilon \searrow 0} \lim\limits_{N \to \infty} \int\limits_{(0,\epsilon]} \mathcal B(E-\mu_N^{\omega}) \,\mathcal N_N^{\omega} (\mathrm{d} E) \\
 & \le \rho^{-1} \lim\limits_{\epsilon \searrow 0} \lim\limits_{N \to \infty} \int\limits_{(0,\epsilon]} \mathcal B(E+\delta) \,\mathcal N_N^{\omega} (\mathrm{d} E) \\
 & \le \rho^{-1} \mathcal B(\delta) \lim\limits_{\epsilon \searrow 0} \lim\limits_{N \to \infty} \int\limits_{\mathds R} g_{\epsilon,\epsilon}(E) \, \mathcal N_N^{\omega} (\mathrm{d} E) \\
& = \rho^{-1} \mathcal B(\delta)\lim\limits_{\epsilon \searrow 0} \int\limits_{\mathds R} g_{\epsilon,\epsilon}(E) \, \mathcal N_{\infty} (\mathrm{d} E) \\ 
& \le \rho^{-1} \mathcal B(\delta) \lim\limits_{\epsilon \searrow 0} \mathcal N_{\infty}^{\mathrm{I}} (2 \epsilon) \\
& = 0
 \end{align*}
 employing $\mathds P$-almost sure vague convergence of $\mathcal N_N^{\omega}$ to $\mathcal N_{\infty}$ and Theorem~\ref{Lifshitz Auslaufer one dimensional}.
\end{proof}

\vspace*{0.5cm}

{\small
\bibliographystyle{amsalpha}
\bibliography{Literature}}

\end{document}